\newif \ifcomments
\commentstrue

\newif\ifdeferproofs      %
\deferproofsfalse        %

\documentclass[11pt, letterpaper]{article}
\usepackage[letterpaper,margin=1in]{geometry}
\usepackage[T1]{fontenc}
\usepackage{graphicx} %
\usepackage{subcaption}
\usepackage{enumerate}
\usepackage{nameref}
\usepackage{todonotes} %
\usepackage{graphicx}
\usepackage{float}
\usepackage{tikz}
\usetikzlibrary{automata,arrows,positioning,decorations.pathreplacing,arrows.meta}
\usepackage{titlesec}
    \titleformat{\subsubsection}[runin] %
      {\normalfont\normalsize\bfseries} %
      {} %
      {0em} %
      {} %
      [] %
\titlespacing*{\subsubsection}
      {0pt} %
      {3.25ex plus 1ex minus .2ex} %
      {1em} %

\usepackage{pgfplots}\pgfplotsset{compat=1.18}
\usepackage[breakable,most]{tcolorbox}
\newtcolorbox{gamebox}[2][]{
    breakable,
    colback=white,
    colframe=black,
    boxrule=0.5pt,
    colbacktitle=black,
    coltitle=white,
    enhanced,
    left=2mm,
    right=2mm,
    top=2mm,
    bottom=2mm,
    after skip=15pt,
    fonttitle=\bfseries,
    attach boxed title to top left={yshift=-0.05in,xshift=0.15in},
    boxed title style={boxrule=0pt,colframe=white,},
    title={#2},
    #1
}
\usepackage[colorlinks,citecolor=blue,linkcolor=blue,urlcolor=blue,pagebackref]{hyperref}
\usepackage[appendix=append, bibliography=common]{apxproof}
\usepackage{amsmath,amsfonts,amsthm,amssymb}

\usepackage[symbols,nogroupskip,numberedsection,sort=none]{glossaries-extra}

\usepackage[capitalise,nameinlink]{cleveref}
\usepackage{etoolbox}
\crefname{appendix}{Appendix}{Appendices}
\Crefname{appendix}{Appendix}{Appendices}
\usepackage{tabularx,array,booktabs,siunitx}
\glsxtrnewsymbol[description={number of players (or units of stake)}]{n}{\ensuremath{n}}
\glsxtrnewsymbol[description={set of all players $\{0,1,\dots,n-1\}$}]{nset}{\ensuremath{[n]}}
\glsxtrnewsymbol[description={proposal delay}]{Delta}{\ensuremath{\Delta}}
\glsxtrnewsymbol[description={MEV reward function $\mu x + \mu_0$}]{MEV}{\ensuremath{M(x)}}
\glsxtrnewsymbol[description={block reward function $b_0 - bx$}]{B}{\ensuremath{B(x)}}
\glsxtrnewsymbol[description={MEV intercept}]{mintercept}{\ensuremath{\mu_0}}
\glsxtrnewsymbol[description={MEV slope}]{mslope}{\ensuremath{\mu}}
\glsxtrnewsymbol[description={block reward intercept}]{bintercept}{\ensuremath{b_0}}
\glsxtrnewsymbol[description={block reward slope}]{bslope}{\ensuremath{b}}
\glsxtrnewsymbol[description={timeout duration}]{dmax}{\ensuremath{\tau}}
\glsxtrnewsymbol[description={maximum delay}]{dstar}{\ensuremath{\Delta^*}}
\glsxtrnewsymbol[description={strategy profile}]{sigma}{\ensuremath{\vec{\sigma}}}
\glsxtrnewsymbol[description={time averaged utility of player $i$ with strategy profile $\sigma$}]{utility}{\ensuremath{U_i(\vec{\sigma})}}
\glsxtrnewsymbol[description={latency from $i$ to $j$}]{lij}{\ensuremath{L_{i \to j}}}
\glsxtrnewsymbol[description={timeliness measurement threshold}]{m}{\ensuremath{m}}
\glsxtrnewsymbol[description={quorum threshold}]{c}{\ensuremath{c}}
\glsxtrnewsymbol[description={strategic proposal delay of the leader (in addition to latencies)}]{delay}{\ensuremath{\Delta}}
\glsxtrnewsymbol[description={quorum time of leader $i$ when previous leader is $j$
($c\ts{th}$ order statistic of $\{L_{j \to k} + L_{k \to i}\}_k$)}]{qt}{\ensuremath{Q^{j \to i}_c}}
\glsxtrnewsymbol[description={average quorum time of leader $i$ ($\frac{1}{n}\sum_{k \in [n]} Q^{i \to k}_c$)}]{avgqt}{\ensuremath{\overline{Q}^i_c}}
\glsxtrnewsymbol[description={earliest proposal duration (starting time) of leader $i$ when previous leader is $j$
(the maximum of quorum duration and block duration $\max\{Q^{j \to i}_c,L_{j \to i}\}$)}]{pt}{\ensuremath{S^{j,i}}}
\glsxtrnewsymbol[description={average starting time of leader $i$ ($\frac{1}{n}\sum_{j \in [n]} S^{j,i}$)}]{avgpt}{\ensuremath{\overline{S}^i}}
\glsxtrnewsymbol[description={honest measurement error of leader $i$ when previous leader is $j$
($x\ts{th}$ order statistic of $\{L_{i \to k} - L_{j \to k}\}_k$)}]{qdiff}{\ensuremath{E^{i-j}_x}}
\glsxtrnewsymbol[description={average honest measurement error of leader $i$ ($\frac{1}{n}\sum_{j \in [n]} E^{i-j}_x$)}]{avgqdiff}{\ensuremath{\overline{E}^{i}_x}}

\newtheoremrep{theorem}{Theorem}[section]
\newtheoremrep{lemma}{Lemma}[section]
\newtheoremrep{definition}{Definition}[section]
\newtheoremrep{claim}{Claim}[section]
\newtheoremrep{proposition}{Proposition}[section]
\newtheoremrep{corollary}{Corollary}[section]
\newtheoremrep{conjecture}{Conjecture}[section]
\newtheoremrep{example}{Example}[section]
\newtheoremrep{property}{Property}[section]
\theoremstyle{remark}
\newtheoremrep{observation}{Observation}[section]
\newtheoremrep{fact}{Fact}[section]

\newcommand{\R}{\mathbb{R}}

\newcommand{\bX}{\overline{X}}
\newcommand{\bY}{\overline{Y}}

\newcommand{\W}{\mathcal{W}}
\newcommand{\eps}{\varepsilon}
\newcommand{\E}{\mathbb{E}}
\newcommand{\leader}{\mathsf{leader}}

\ifcomments
    \newcommand{\kushal}[1]{{\color{red}{[KB: #1]}}}
    
    \newcommand{\kaya}[1]{{\color{green}{[KA: #1]}}}
    \newcommand{\aditya}[1]{{\color{magenta}{[AS: #1]}}}

\else
    \newcommand{\kushal}[1]{}
    
    \newcommand{\kaya}[1]{}
    \newcommand{\aditya}[1]{}
\fi

\newcommand*{\email}[1]{%
    \footnotesize\href{mailto:#1}{#1}\par
    }
\title{Timing Games in Responsive Consensus Protocols}
\author{Kaya Alpturer\thanks{E-mail: \email{kalpturer@princeton.edu}. Work done while interning at Category Labs.}\\[0.3em]Princeton University \and Kushal Babel\thanks{E-mail: \email{babel@cs.cornell.edu}.}\\[0.3em]Category Labs \and Aditya Saraf\thanks{E-mail: \email{as2777@cornell.edu}. Work done while interning at Category Labs.}\\[0.3em]Cornell University}
\date{}
\newcommand{\ts}{\textsuperscript}

\begin{document}

\maketitle
\begin{abstract}
	Optimistic responsiveness---the ability of a consensus protocol to operate at the speed of the network---is widely used in consensus protocol design to optimize latency and throughput.
	However, blockchain applications incentivize validators to play \textit{timing games} by strategically delaying their proposals, 
    since increased block time correlates with greater rewards.
	Consequently, it may appear that responsiveness (even under optimistic conditions) is impossible in blockchain protocols.
    In this work, we develop a model of timing games in responsive consensus protocols  
    and find a prisoner's dilemma structure, where 
    cooperation (proposing promptly) is in the validators' best interest, but individual incentives encourage validators 
    to delay proposals selfishly.    
    To attain desirable equilibria, 
    we introduce dynamic block rewards
    that decrease with round time to explicitly incentivize faster proposals. 
    Delays are measured through a \textit{voting} mechanism, where other validators vote on the current leader's round time. 
    By carefully setting the protocol parameters, the voting mechanism allows validators to coordinate and reach the cooperative equilibrium, benefiting all through a higher rate-of-reward.
	Thus, instead of responsiveness being an unattainable property due to timing games,
	we show that responsiveness itself can promote faster block proposals.
    One consequence of moving from a static to dynamic block reward 
    is that validator utilities become more sensitive to latency,
    worsening the gap between the best- and worst-connected validators.
    Our analysis shows, however, that this effect is minor in both 
    theoretical latency models and simulations based on real-world networks.
    
\end{abstract}

\section{Introduction}
A major focus in consensus protocol design for Proof-of-Stake blockchains has
been designing protocols that are \emph{optimistically responsive}~\cite{YMRGA19}. 
In the happy path, these protocols proceed at the speed of the network, 
whereas non-responsive protocols always wait for the worst case network delay.
Responsiveness helps the protocol adapt to network conditions, 
by making rounds move faster than in non-responsive protocols like Ethereum, 
which have a fixed round time.
Adaptive round times benefit from lower latency and higher throughput.%
However, responsive protocol design is complicated by the fact that the blockchain
setting has novel incentive structures that highly depend on the timing of the
validators' actions. 
One of the major sources of reward for validators is maximal extractable value (MEV)~\cite{babel2023clockwork,daian2020flash}, 
the additional value the validators are able to extract by building their blocks strategically.
Examples of MEV include the profit that validators make from arbitrage between exchanges
(centralized and decentralized), frontrunning or backrunning user transactions, 
transaction fees, and so on.
As MEV tends to increase with proposal time \cite{OKVKMT23},
a selfish leader will delay their block proposal as long as possible, even risking a timeout. 
This phenomenon has been called a \textit{timing game}
\cite{SSTPSM23}, and has been observed on blockchains such as Ethereum~\cite{SSTPSM23} 
and recently on Solana~\cite{solana-slowpokes,solana-chorusone-timing}.

Timing games are undesirable as they increase latency to users, and increase the likelihood of missed blocks, reducing throughput even in non-responsive protocols. 
Timing games pose an even greater risk in blockchains relying on responsive protocols, where slower round times would invalidate the responsiveness offered by the underlying consensus protocol.
A natural question then is the following:
\textbf{Is optimistic responsiveness possible in blockchain protocols when there
 are strong incentives to delay blocks due to MEV?}

In this paper, we propose explicitly rewarding timely block production to shift
incentives against timing games. If faster block proposals are assigned greater
block reward, the reward can counterbalance MEV.

Implementing block rewards that depend on proposal duration is not trivial. One
challenge is that there is no reliable measure of time available to the
protocol. For instance, relying on block timestamps alone is not possible as
they can be strategically manipulated \cite{YSZ23}.
To measure proposal durations accurately, we introduce a \emph{timeliness
vote}. All validators report their observed durations, and the
reported durations are aggregated to form a proxy to the real duration. However, 
strategic validators may misreport their observed time, either unilaterally or as 
part of a coalition.

Our results suggest that responsiveness can induce honest voting, 
which disincentivizes timing games. 
This is because timing games often take the form of a prisoner's dilemma among validators. 
Validators can \emph{cooperate} by proposing early, which
increases the speed of the system, and hence the rate of block reward for everyone. 
Alternatively, validators can defect by proposing late, greedily gaining some MEV while slowing down the system.
Since each validator is a small part of the network, 
proposing late is a dominant strategy in the absence of coordination 
(i.e., the standard prisoner's dilemma result). Honest voting, along 
with time-decreasing block reward, allows validators to coordinate to 
reach the early-proposing equilibrium rather than the late-proposing equilibrium.
While it is clear that MEV allows validators to profit from users in the single-block horizon, 
our work highlights the surprising fact that MEV can hurt validators in the long-run by slowing down the rate of block reward.

Finally, we consider a notion of \textit{fairness} that models disparities in rewards to validators due to heterogeneous latencies. 
Being better connected to other validators is helpful under both dynamic and static block rewards. Unfortunately, this fact encourages geographical 
collocation, which compromises decentralization. We show that, while fairness 
does degrade with time-decreasing block rewards, the gap is small, in a variety of theoretical and empirical
settings.

\subsection{Overview of results}
We develop \textbf{a model of timing games in optimistically responsive BFT
protocols}.
Diverging from prior work, we explicitly incorporate the randomness
of peer-to-peer latencies.
Modeling randomness ensures that validators cannot distinguish intentional
delays from genuinely higher latency.

\subsubsection{Early-proposing equilibria.}
We show that if the block reward is set appropriately and individual entities have bounded stake, early-proposing and honest-reporting of
timeliness is a Nash equilibrium (\cref{thm:latency-NE}). 
We then argue that this equilibrium is robust in several ways. We first consider ``whales'', i.e., single validators who control many units of stake, and show that they won't deviate as long as they are sufficiently small (e.g., controlling less than $1/3$ of the total stake under standard parameter choices). 
We also argue that rational coalitions are unstable. Assuming a low-latency condition, we prove that a certain natural family of ``large'' coalitions will revert back to the desirable equilibrium (\cref{thm:large-coalition}). We also argue that smaller coalitions are unstable because members would be incentivized to join many small coalitions, to the detriment of others.

\subsubsection{Fairness.}
We analyze fairness with respect to latencies.
We first theoretically study two models of latency: one where nodes are arranged on a line and
one where they are clustered into two groups, and show that fairness gets
marginally worse under dynamic block rewards compared to static block rewards (\cref{thm:fairness-line,thm:fairness-cluster,thm:fairness-late}). We then show similar results in simulations of real-world latency data.
Our results suggest that there is some fairness cost to timing game
mitigation. 
\subsubsection{Alternative approaches.}
An alternative incentive knob to decreasing block reward is to decrease leader election probabilities as a function of their round duration. However, it's more difficult to incentivize honest voting in that setting, as a decrease in others' election probability directly increases one's own.

\subsection{Related work}
\subsubsection{Timing incentives.}
Timing games and incentives for delaying one's block in order to propose a more
rewarding block have been both empirically observed~\cite{solana-slowpokes,SSTPSM23,solana-chorusone-timing} and theoretically analyzed~\cite{SSTPSM23}.
Besides modeling responsive protocols, our model diverges from the approach of
\cite{SSTPSM23} by modeling peer-to-peer latencies.
The increase in MEV with time has been empirically measured by \cite{OKVKMT23,WZQSG23}.
Solutions to timing games have been explored as well \cite{C22,SXNC25}.
The blog post \cite{C22} proposes rewarding timeliness using a 
proxy measurement of same-slot committee votes in Ethereum. 
\cite{SXNC25} proposes measuring transaction
arrival rates and correcting timing incentives by redistributing execution layer
fees. Our setting, however, requires analyzing incentives explicitly, as the effect of deviations directly affects the protocol's operation.

\subsubsection{Time-rotating leaders.}
A recent proposal by \cite{KGMS25} involves modifying consensus so that leaders are elected for a certain period of time, rather than a round (or number of rounds). Here, even with a constant block reward, leaders may prefer to propose as fast as possible, as that allows them to propose more often. This is a promising approach for ensuring timely inclusion of \textit{non-MEV} transactions, such as direct transfers. 

However, for the transactions that generate MEV, such as decentralized finance (DeFi) transactions, the leader would likely simply hold onto those transactions and only submit them in their final block. Doing so would allow them to better optimize MEV in the last period. For a specific example, the arbitrage profit between a centralized exchange and a decentralized exchange over a certain time range $T$ grows proportional to $T^{3/2}$ \cite{MMR24}. Thus, the leader may prefer to wait until their final block to include DeFi transactions.

\subsubsection{Responsive Consensus.}
There is a vast literature on fault-tolerant consensus protocol design that
utilizes responsiveness but does not consider timing incentives.
Responsiveness was first introduced in \cite{ADLS94},
and optimistic responsiveness with the widely used fast-vs-slow path framework
was introduced in \cite{PS18}. Numerous BFT consensus protocols were designed focusing on
blockchain settings with responsiveness in mind
\cite{GKSSX22,GSDHAC23,JB25,YMRGA19}.
The tendency of leaders to slow down has been acknowledged in \cite{KGMS25},
and punishing misses by changing the leader election schedule has been explored
by \cite{CGKLMSS22,TKSK24}.
However, as far as we are aware, these  works do not consider timing games
from a game-theoretic perspective and the proposed solutions are not analyzed
in a model where the validators are rational.

\subsubsection{Blockchain protocol incentives.}
In addition to timing games, there is a growing body of work investigating
incentive incompatibilities in blockchain protocols---including selfish
mining in Bitcoin \cite{ES14,SSZ16}, Proof-of-Stake attacks \cite{AW24,BNPW19}, 
and timestamp manipulation \cite{YSZ23,YTZ22}.
Moreover, MEV has been observed to interfere with performance and scaling
\cite{mevlimitscaling}. Lastly, our results show that when block rewards are too
low, timing games are inevitable. A similar observation was made for Bitcoin's
block reward in relation to forking attacks in \cite{CKWN16}.

\subsection{Roadmap}
In \cref{sec:preliminaries}, we give a brief summary of the optimistically responsive consensus protocols 
that we consider and the solution concepts we use.
In \cref{sec:model}, we formalize timing games in responsive protocols with time-dependent block reward. In \cref{sec:early-equilibria}, we present our equilibria analysis. \cref{sec:fairness} analyzes the fairness for both early and late-proposing scenarios in several 
models of network latencies.
Lastly we give guidance on parameter settings in~\cref{sec:reward} and conclude with a discussion of future work in \cref{sec:discussion}. 

Proofs are deferred to \cref{apx:proofs-equilibria,appendix:coalition_resistance,apx:fairness}.
\cref{appendix:sim} has simulations for the fairness models.
\cref{appendix:inflation} discusses utility modeling.
\cref{apx:leader} explores an alternative approach of modifying leader election probabilities as the incentive knob instead of block reward. For a glossary of notation, see \cref{apx:glossary}.

\section{Preliminaries} \label{sec:preliminaries}
\subsubsection{Optimistically responsive Consensus.}
We adopt the rotating leader framework of optimistically responsive consensus 
protocols (such as \cite{YMRGA19,JB25}) where validators vote on a proposed block by the leader, and the next leader aggregates a certain threshold number of votes to form a ``quorum certificate''.
 In the happy path where the quorum is able to be formed, 
the protocol makes progress as soon as messages arrive, rather than waiting for the worst case time bound.
If the nodes do not hear from the leader within a certain duration, a \textit{timeout} is triggered to proceed with the next leader.
The responsive paradigm is in contrast to fixed-duration rounds. For instance, Ethereum blocks are 12 seconds 
regardless of how fast the proposal and vote messages are propagated. The same is true for Tendermint~\cite{B16} as well.
As we analyze the game from a rational model, we will mostly focus on the common case of happy-path operation of the protocol. For leader election, we assume a leader is randomly sampled each round, as is the case for Proof-of-Stake blockchain protocols.

\subsubsection{Subgame perfect Bayes-Nash equilibria.}
We are modeling the  blockchain validators as strategic participants in an infinitely 
repeated game. Therefore, the solution concept we use is subgame perfect Bayes-Nash equilibria.
A strategy profile $\vec{\sigma}$, which specifies the voting and proposal strategy of each validator,
is a subgame perfect Bayes-Nash equilibrium if, in any possible subgame, 
the validator is best-responding with respect to its expected utility by utilizing $\vec{\sigma}$.

\subsubsection{Coalition resistance.} We will also consider coalitions best-responding to other validators in 
a strategy profile. While we discuss the types of coalition we consider in \cref{sec:model}, 
the key condition we require is coalition-resistance: no coalition (satisfying some condition such as size)
can deviate from honest behaviour (early proposing and honest voting) and strictly get more utility in expectation.

\section{Model} \label{sec:model}
Our model abstracts away many complexities of running consensus among a distributed set of validators in order to 
capture a wide range of protocol implementations. 
One key complexity that we do retain is the existence of pairwise latencies among validators and 
we are concerned with the time messages take from one validator to another. Whether the delivery is through a gossip network
or direct communication is abstracted away. 
We first explain the model basics, then present the full game, and finally discuss agent utilities in more detail. 

\subsubsection{Setup.}
In our model, we consider the game induced by introducing a block reward
that depends on timeliness of proposals. There are $n$ players $\{1, \dots, n\} = [n]$ participating in the 
protocol.\footnote{Each agent is a uniform unit of stake. 
We will consider the effect of heterogeneous stake distributions and collusion via coalitions.}
A leader is elected uniformly randomly each round and, since the protocol is responsive, may propose a block as soon as they have received the previous block and a quorum of $c$ validators' votes (the quorum threshold $c$ is $\lfloor\frac{2n}{3}\rfloor+1$ for most protocols). A strategic leader may choose to intentionally delay its block by some delay $\Delta$.
Validators measure the {\it round duration} locally as the time elapsed on their clock between receiving the previous leader's proposal and the current leader's block proposal. 
$\tau$ represents the timeout duration after which honest validators time out on the current leader's proposal. In other words, if the round duration exceeds $\tau$ for an honest validator, it will timeout.
Validators move to the next round after a quorum (also $c$-sized for simplicity) has either timed out, or confirmed the current leader's proposal.
Finally, $L_{i \to j}$ represents the latency distribution of messages from $i$ to $j$.\footnote{While we do not differentiate between block propagation 
latency and voting latency, they are often different in real life protocols, and it is straightforward to introduce separate distributions for both 
in our results. }

\subsubsection{Timeliness measurement.}
We propose a novel timeliness measurement that does not exist in current consensus implementations.
Upon receiving a proposal, each validator votes whether to approve the block, and additionally includes the 
observed proposal duration (their vote is implicitly infinite if not reported). The measurement time is then defined as the $m\ts{th}$ smallest vote, i.e., the $m\ts{th}$ \textit{order statistic} of the votes. We assume votes are aggregated honestly.\footnote{One possible implementation would be to require that timing votes are directly submitted on-chain by validators, and aggregated over a sufficiently long window spanning a number of leaders. This allows one to sidestep any strategic aggregation issues that would arise if the next leader is responsible for aggregation.}

\subsubsection{Optimal-delaying oracle, bounded latencies, and mostly honest timeouts.}
We assume that the validators are able to delay maximally without missing their block. We model this with the existence of an oracle that sees 
all the latency realizations and informs the leader of the largest delay that ensures their block doesn't timeout.
This assumption does not map to practice, but it only
makes our results stronger, as we will show that validators prefer to propose early 
despite this ``timeout protection''.
Note that this oracle depends on the voting strategy of validators, as some validators may timeout dishonestly. 

As this maximum delay varies in each round due to the latency realizations, let $\Delta_r^*$ 
be the maximum delay in round $r$. We drop the subscript when it is clear from context. 
To ensure that all validators can participate, we require $\Delta^* \ge 0$. 
In other words, if a leader does not intentionally delay, their proposal will not timeout. 
This is equivalent to assuming (1) a bounded latency condition, (2) a sufficiently long timeout, and (3) a limited number of validators will dishonestly time out. Condition (3) essentially requires that $c$ validators will not \textit{timeout} dishonestly, though they may vote dishonestly otherwise. 
We know that $\Delta^*_r \le \tau$, since latencies are nonnegative. Thus, $\tau$ should be thought of as the maximum possible delay over all latency realizations.
While there may exist an incentive for large coalitions to censor other validators (i.e., 2/3 of the network can censor out the other 1/3), or benefit their coalition members by extending the timeout, these issues already exist without timing games, and our work does not aim to address them. Instead, we focus on incentives regarding the timeliness vote.

\subsection{Formal game}
We now present the formal model of the game. We present a general round, and then note the special differences that apply to the first round.
For most of our variables, we denote the variable's value in round $R$ by a subscript.
\begin{figure}[htp]
    \centering
    \begin{tikzpicture}[
            every node/.style = {align=center},
            Line/.style = {-angle 90, shorten >=2pt},
            Linel/.style = {-angle 90, shorten >=4pt},
            Brace/.style args = {#1}{semithick, decorate, decoration={brace,#1,raise=2pt,
                                     pre=moveto,pre length=2pt,post=moveto,post length=2pt,}},
            Bracel/.style args = {#1}{semithick, decorate, decoration={brace,#1,raise=20pt,
                                     pre=moveto,pre length=2pt,post=moveto,post length=2pt,}},
            ys/.style = {yshift=#1},
            rl/.style = {raise=#1}
        ]
        \linespread{0.9}                    
        \coordinate (a) at (0,0);
        \coordinate[right=22mm of a]    (b);
        \coordinate[right=13mm of b]    (c);
        \coordinate[right=14mm of c]    (d);
        \coordinate[right=38mm of d]    (e);
        \coordinate[right=16mm of e]    (f);
        \coordinate[right=5mm of f]    (g);
        
        \draw[Line] (a) -- node[above] {}    (b) -- (f)
                        -- node[above] {}         (g) node[below left] {time};        
    
        \draw[Line] ([ys=11mm] e) node[above] {$t_R$} -- (e);
        \draw[Line] ([ys=11mm] b) node[above] {$t_{R-1}$} -- (b);
        \draw[Line] ([ys=7mm]  d) node[above] {quorum received by $i$} -- (d);
        
        \draw[Line] ([ys=7mm] e) node[above] {} -- (e);
        \draw[Line] ([ys=7mm] b) node[above] {} -- (b);
        \draw[Brace] (b) -- node[above=3pt] {$l_{j \to i}$} (c);
        \draw[Line] ([ys=7mm] d) node[above] {} -- (d);
        \draw[Brace=mirror] (d) -- node[below=3pt] {$\Delta$} (e);
        \draw[Brace=mirror] (b) -- node[below=3pt] {$s_{R-1}$} (d);

    \end{tikzpicture}
    \caption{Timeline: $t_R$ represents the time at which 
    the leader of round $R$ broadcasts their block. 
    The leader of round $R-1$ is $j$ and the leader of round $R$ is $i$.}
\end{figure}
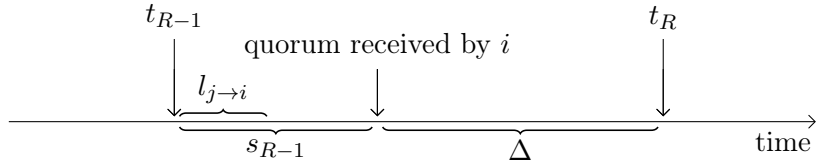
\begin{gamebox}[]{timing game}
\begin{enumerate}
    \item Let $R$ be the current round, and suppose $i$ is the current leader and $j$ led the previous round. 
    Let $t_{R-1}$ represent the time when $j$ broadcasts its block. 
    Let $s_{R-1}$ denote the duration between when $j$ broadcasts their block and when $i$ receives the block and its quorum. 
    \item The leader $i$ learns their maximum delay $\Delta^*_R$
    and chooses a delay $\Delta \in [0, \Delta^*_R]$.
    \item Let $k \neq i$ be an arbitrary validator. 
        \begin{enumerate}
            \item The latency $l_{i \to k} \sim L_{i \to k}$ is sampled.
            \item $k$ receives the block at time $t^{k}_{R} = t_{R-1} + s_{R-1}+ \Delta + l_{i \to k}$. 
            \item $k$ votes according to their voting strategy, which is a function of the block times they have seen. 
            So $v^k_R =\sigma^V_k(\mathbf{t}^k_{1:R}) \in \mathbb{R}^+$. $k$ then publishes
            its vote
            where $\sigma^V_k$
            is the voting strategy of validator $k$.
        \end{enumerate}
    \item Let $v^*$ refer to the aggregated timeliness vote ($m\ts{th}$ smallest vote in the (multi)set $\{v_R^k\}_k$). $i$ receives utility $M(\Delta+s_{R-1}) 
    + B(v^*)$, where $M$ denotes the MEV and $B$ the block reward.
    \item The next leader is randomly sampled, and the game repeats from 1.
\end{enumerate}
The first round of this game is slightly different. The game starts at time $t = 0$, 
so $s_{0} = 0$, i.e., the leader $i$ can propose as soon as the game starts. 
As a result, validator $k$ receives the block at time 
$t_1^k = \Delta + l_{i \to k}$. Otherwise, the first round matches the others.
\end{gamebox}

\subsubsection{Assumptions.} We highlight some additional assumptions for our model. 
\begin{enumerate}[(a)]
    \item \emph{Clock synchronization.} We do not deal with clock drift in this model. 
    We assume that all validators agree on the start time of the first round and that their clocks tick at the same rate. 
    \item \emph{Hidden latency realizations.} While we assume the latency distributions are public, 
    the realized latencies are unknown. So, validators don't know exactly when the previous leader proposed, 
    as they can't distinguish the low-latency with delayed-proposal case from the high-latency with early-proposal case. They also don't know exactly when other validators receive a block. 
    \item \emph{Efficient block building.} 
    Everyone proposes every round and extracts all potential MEV (given by function $M(\cdot)$) at the time of proposal.\footnote{There may be incentives for validators to coordinate across the reigns of different leaders and exclude MEV-generating transactions for several rounds before finally claiming and sharing the MEV. Note that this strategy already exists even with static block rewards and is not empirically observed, suggesting practical coordination barriers.}
\end{enumerate}

\subsection{Utilities and coalitions}
\subsubsection{Rewards.} There are several sources of reward that drive the utilities of validators in the system.
The main source of timing games is the MEV, which increases for blocks that are built later in time due to an increase in arbitrage opportunities, more time to grind through ordering combinations for the given transactions, and a larger mempool (available user transactions).
The protocol, on the other hand, assigns reward for proposing a block to the proposer and voter rewards for block validation. We model MEV as an increasing function of the time between block proposals. We ignore voting rewards as they are constant.
\begin{itemize}
\item \emph{MEV reward:} $M(x) : \R \to \R$ is a positive increasing function. In the fairness analysis, we will assume that it is of the form $\mu x + \mu_0$ for convenience.
\item \emph{Block reward:} $B$ is chosen by the protocol designer and can be a function of 
information available to the protocol such as the timeliness vote. We design a linearly decreasing block reward that reaches $0$ when the timeliness vote is equal to $\tau$. Thus, $B(x) = b_0 - bx$ for $b \geq 0$, where $b_0 = b\tau$.
\end{itemize}
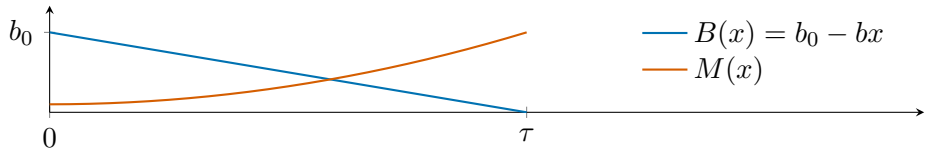
\begin{figure}[ht]
    \centering
    \begin{tikzpicture}
      \begin{axis}[
        width=0.8\textwidth, height=3cm,
        axis lines = left,
        xmin = 0, xmax = 5.5,
        ymin = 0, ymax = 4,
        xtick = {0, 3},
        xticklabels = {$0$, $\tau$},
        ytick = {3},
        yticklabels = {$b_0$},
        legend pos = north east,
        legend style = {draw=none, fill=none},
        legend cell align = {left}
      ]
        \addplot[domain=0:3,color={rgb,255:red,0;green,114;blue,178}, thick] {3-x};
        \addlegendentry{$B(x) = b_0 - bx$}
    
        \addplot[domain=0:3,color={rgb,255:red,213;green,94;blue,0}, thick] {0.3+0.3*x*x};
        \addlegendentry{$M(x)$}
      \end{axis}
    \end{tikzpicture}    
    \caption{Reward functions.}
    \label{fig:rewards}
\end{figure}

\subsubsection{Long-run time-averaged utility.} 
Each player $i$ has utility equal to their reward per unit of time, in the long-run. 
For each round $r$, let $\Delta_r$ denote the delay chosen by the leader of that 
round $\leader_r$, and let the random variable $S_{r-1}$ denote the duration between the leader
of round $r-1$ proposing their block to the time when the current leader can
start proposing, which is a function of the latency distribution. 
The duration $\Delta_r+S_{r-1}$ determines the MEV reward that the leader receives,
as well as the duration that round lasts. 
Let $v^*_r$ denote the timeliness vote in round $r$. 
Then, player $i$'s utility is:\footnote{An astute reader may notice that we are not modeling inflation, which is relevant because players, through their voting, have some control over how much block reward is printed. However, inflation does not qualitatively change any of our results. We defer a longer discussion to \cref{appendix:inflation}.}

\begin{align*}
    \E \left[
    \lim_{R \to \infty}
    \frac{\sum_{r=0}^{R}\mathbb{I}[i = \leader_r] 
    (M(\Delta_r+S_{r-1}) + B(v^*_r))}
    {\sum_{r=1}^R \Delta_r + S_{r-1}}
    \right]
\end{align*}

\subsubsection{Coalitions.} 
We consider two types of coalitions: (1) a single entity (``whale'') controlling multiple units of stake; (2) separate agents cooperating in a group.
\begin{itemize}
    \item \emph{Whale:} 
    A whale $C \subseteq [n]$ is controlled by a single entity, which is only concerned with the joint utility of all nodes in $C$. So, a strategy profile $\sigma$ is $\alpha$-whale proof if, for any $C \subseteq [n]$ such that $|C|/n \le \alpha$, there exists no deviation that improves the joint utility of $C$. 
    Note that for all $i,j \in C$, $L_{i, j} = 0$. In other words, whales' stake is assumed to be co-located. 
    
    \item \emph{Rational coalition:} 
    Let $C \subseteq [n]$ denote a rational coalition among separate agents. Here, we require that any suggested deviation is itself a Bayes-Nash equilibrium, and we consider each agent's utility.     In other words, a strategy profile $\sigma$ is $\alpha$-rational coalition proof if, for any deviation $\sigma'$, one of the following two conditions holds. 
    \begin{enumerate}[a]
        \item Some $i \in C$ has higher utility when everyone plays $\sigma$ as opposed to the coalition playing $\sigma'$ while the rest play $\sigma$.
        \item Some $i \in C$ prefers to deviate from $\sigma'$, given that the rest of the coalition plays $\sigma'$ and the non-coalition validators play $\sigma$.
    \end{enumerate}
\end{itemize}

\section{Early-proposing truthful equilibria} \label{sec:early-equilibria}
Let $\sigma^*$ be the truthful strategy profile where each leader proposes without delay ($\Delta = 0$) and voters honestly report their observed round duration. In this section, we investigate when $\sigma^*$ is an equilibrium. Full results are deferred to \cref{apx:proofs-equilibria}.

\subsection{Honest equilibria}
We first define $S^{j, i}$, the random variable that represents the time from 
when current leader $i$ can start proposing, given that $j$ led the previous round. 
For some validator $k$, the time it takes for them to receive $j$'s block and 
send a vote to $i$ is distributed according to $L_{j \to k} + L_{k \to i}$. 
For a quorum size of $c$, let $Q^{j \to i}_c$ denote the $c\ts{th}$ order statistic on 
$\{L_{j \to k} + L_{k \to i}\}_k$. This is the exact time when $i$ has enough votes to form a quorum. 
Since $i$ also needs the block to build on it, $S^{j, i} = \max(Q^{j \to i}_c, L_{j\to i})$, 
where the max is necessary because it's possible for $i$ to receive the quorum 
before they have seen $j$'s block. 

If $i$ delays by $\Delta \ge 0$, 
they extract $M(\Delta + S^{j,i})$ due to MEV. 
For the duration votes, note that validator $k$'s honest report, 
given that $i$ delays by $\Delta$, is distributed according to $V^k = \Delta + S^{j ,i} + L_{i \to k} - L_{j \to k}$.
Aggregating at least $m$ of these votes, the block reward of leader $i$ for a delay of $\Delta$ 
is distributed according to $B(\Delta+S^{j, i} + E^{i-j}_{m})$
where $E^{i-j}_m$ represents the $m\ts{th}$ order statistic of $\{L_{i \to k} - L_{j \to k}\}_k$. 
Intuitively, $E^{i-j}_{m}$ represents the error in measuring the leader's 
duration by the timeliness vote due to latency in the observations.

The honest proposal strategy is to set $\Delta = 0$, so leader $i$ would receive total reward 
$M(S^{j, i}) + B(S^{j, i} + E^{i-j}_{m})$. Lastly, observe that
$\Delta^* = \tau - S^{j,i} - E^{i - j}_c$ for leader $i$ preceded by $j$.\footnote{Note 
that we now use the $c\ts{th}$ order statistic as the quorum threshold is $c$ as $\Delta^*$
is determined by $c$ validators not timing out.}
Figure~\ref{fig:honest-voting} illustrates the timeline of voting honestly.
\begin{figure}[htp]
    \centering
    \begin{tikzpicture}[
            every node/.style = {align=center},
            Line/.style = {-angle 90, shorten >=2pt},
            Linel/.style = {-angle 90, shorten >=4pt},
            Brace/.style args = {#1}{semithick, decorate, decoration={brace,#1,raise=2pt,
                                     pre=moveto,pre length=2pt,post=moveto,post length=2pt,}},
            Bracel/.style args = {#1}{semithick, decorate, decoration={brace,#1,raise=18pt,
                                     pre=moveto,pre length=2pt,post=moveto,post length=2pt,}},
            ys/.style = {yshift=#1},
            rl/.style = {raise=#1}
        ]
        \linespread{0.9}                    
        \coordinate (a) at (0,0);
        \coordinate[right=22mm of a]    (b);
        \coordinate[right=16mm of b]    (c);
        \coordinate[right=16mm of c]    (d);
        \coordinate[right=24mm of d]    (e);
        \coordinate[right=16mm of e]    (f);
        \coordinate[right=22mm of f]    (g);
        
        \draw[Line] (a) -- node[above] {}    (b) -- (f)
                        -- node[above] {}         (g) node[below left] {time};
        \draw[Brace] (b) -- node[above=3pt] {$l^{r-1}_{j\to k}$} (c);
        \draw[Brace] (e) -- node[above=3pt] {$l^r_{i \to k}$} (f);
        
        \draw[Brace=mirror] (b) -- node[below=3pt] {proposal duration} (e);
        \draw[Bracel=mirror] (c) -- node[below=20pt] {observed duration by $k$} (f);

        \draw[Line] ([ys=11mm] e) node[above] {block proposed by $i$} -- (e);
        \draw[Line] ([ys=11mm] b) node[above] {block proposed by $j$} -- (b);
        \draw[Line] ([ys=7mm]  c) node[above] {received by $k$} -- (c);
        \draw[Line] ([ys=7mm]  f) node[above] {received by $k$} -- (f);
    \end{tikzpicture}
    \caption{Round timeline for honest voting. $i$ is the current leader, $j$ is the previous leader, and $k$ is some validator observing $i$'s duration.
    $l^r_{i \to k}$ represents the realized latency in round $r$ between $i$ and $k$.}\label{fig:honest-voting}
\end{figure}
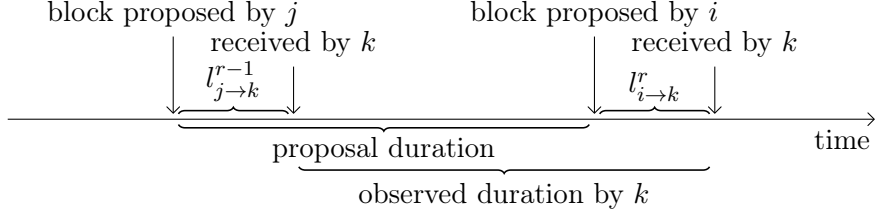

Note that some votes may be negative; due to latency, validator $k$ may receive the block from $i$ \textit{before} the block from $j$. If the $m\ts{th}$ smallest vote was negative, then the utility under $\sigma^*$ would not be well-defined, as $B$ is defined on the interval $[0, \tau]$. We first identify when this is not the case.

\begin{toappendix}
\section{Equilibria Results (\cref{sec:early-equilibria})} \label{apx:proofs-equilibria}
\end{toappendix}
\begin{lemmarep}\label{lm:positive_median}
    If $m > n - c$, then the $m\ts{th}$ smallest vote is always at least the round delay $\Delta$
    in the early-proposing truthful strategy profile $\sigma^*$.
\end{lemmarep}
\begin{proof}
    Let $i$ be an arbitrary current leader and $j$ the previous leader in round $r$. If a validator $k$ contributed to the quorum in $i$'s proposal, they must have started their round-$r$ clock before $i$ receives the quorum. Thus, if $i$ delays by $\Delta$, $k$'s honest vote would be above $\Delta$. Since $m > n -c$, the $m$th smallest vote will come from one of the above quorum-contributing validators, which means that the aggregated vote is at least $\Delta$.
\end{proof}

We assume $m > n-c$ from now on. 
This ensures that the aggregated timeliness vote will always come from a validator who contributed to the quorum.
We show that, under a certain design of the block reward function, 
$\sigma^*$ (the early-proposing truthful strategy profile) is a 
subgame-perfect Bayes-Nash equilibrium. The key condition is the following.
\begin{definition}[time-decreasing] \label{def:time-decreasing}
    We call the reward functions $\langle M, B \rangle$ time-decreasing
    if $M(x) + B(x)$ is a decreasing quantity in $x$. 
\end{definition}
The condition guarantees that the per-round reward is maximized by proposing at each validator's earliest 
possible time.
Now, we show that $\sigma^*$ is an equilibrium. Intuitively, if rewards are time-decreasing, delaying only hurts $i$'s utility. For voting, note that $i$ can only vote on other leader's proposals, and thus $i$'s voting strategy will not affect their own utility. 
\begin{theoremrep} \label{thm:latency-NE}
    Suppose $\langle M, B \rangle$ are time-decreasing and $m > n-c$. Then, the early-proposing truthful strategy profile $\sigma^*$ is a subgame-perfect Bayes-Nash equilibrium, where the time-averaged utility of $i$ is:
    \[
        U_i(\sigma^*) = 
        \frac{\sum_j \left( \E\left[M(S^{j, i})\right] + \E\left[ B(S^{j, i} + E^{i-j}_{m}) \right] \right)}
         {\sum_{j, k} \E[S^{j, k}]}
    \]
\end{theoremrep}
\begin{proof}
    Consider an arbitrary subgame: $i$ is the leader, and $j$ is the previous leader. We will show that $i$ never wants to propose a delay $\Delta \in (0, \Delta^*]$ or vote dishonestly.

    First, note that $i$'s own utility does not depend on their voting strategy, as their votes only determine the block reward when other validators are the leader. So, we focus on proposal delays.
    When $i$ is proposing their block, they have 
    observed the time they received the block and the time they 
    received the quorum. Since all other validators follow $\sigma^*$, $i$ knows that $j$ did not intentionally delay 
    their block. Thus, by subtracting these times, they can learn 
    $q_c^{j \to i}-l_{j \to i}$, where $q_c^{j \to i} \sim Q_c^{j 
    \to i}$ and $l_{j \to i} \sim L_{j \to i}$. Depending on the 
    latency distributions, this could reveal some information 
    about when validators have received block $j$ and started 
    their clock for $i$'s round. Let $\tilde{E}^{i - j}_{m}$ 
    represent the distribution of the $m\ts{th}$ latency 
    difference, conditional on $i$'s observations. Similarly, let $
    \tilde{S}^{j, i}$ represent the distribution of $i$'s start 
    time, conditional on their observations. 

    From \cref{lm:positive_median}, we know that $\Delta +\tilde{S}^{j, i} + \tilde{E}^{i - j}_{m} \ge \Delta > 0$. From the optimal-delaying oracle, we also know that $\Delta +\tilde{S}^{j, i} + \tilde{E}^{i - j}_{m} \le \tau$. Thus, 
    $B(\Delta +\tilde{S}^{j, i} + \tilde{E}^{i - j}_{m})$ is well defined. 
    Now, $i$'s expected reward for the current round for delay $\Delta$, given their observations, is:
    \begin{align*}
        \E[\textsc{Reward}_i^{(r)} \mid j,r-1] 
        &= \E\left[M(\Delta +\tilde{S}^{j, i}) + B(\Delta + \tilde{S} + \tilde{E}^{i-j}_{m}) \right] \\
        &= \E[M(\Delta + \tilde{S}^{j, i}) + B(\Delta+\tilde{S}^{j, i})] - b\tilde{E}_{m}^{i-j}
    \end{align*}
    Where $b$ is the block reward slope.
    
    Since $\langle M, B \rangle$ are time-decreasing, $\Delta = 0$ maximizes the expectation, regardless of the realization $s^{j, i} \sim S^{j,i}$, and thus maximizes the overall reward, regardless of $i$'s information. Thus, $\sigma^*$ is a subgame-perfect Bayes-Nash equilibrium.

    We apply the Renewal-Reward theorem to compute the time-averaged expected utility of $i$ in $\sigma^*$. That allows us to first compute the expected reward of $i$ in a single round, and then the expected duration of a round, and then divide. For the reward, note that the leaders are elected uniformly at random. Thus,
    consider the expected reward of $i$ in a given round:
    \begin{align*}
        \E[\textsc{Reward}_i] 
        &= \Pr[i = \leader]\sum_j \Pr[j = \leader_\text{prev}]
        \E[\textsc{Reward}_i \mid j = \leader_\text{prev}] \\
        &= \frac{1}{n}\sum_j \frac{1}{n} (\E[M(S^{j, i})] + \E\left[B(S^{j, i} + {E}^{i-j}_{m}) \right]) \\
        &= \frac{1}{n^2} \sum_j \E[M(S^{j, i})] + B\left(\E\left[S^{j, i}\right]\right) 
        -b\left(\E\left[{E}^{i-j}_{m}\right]\right)
    \end{align*}
    Now, note that since $\Delta = 0$, the expected duration of a round where $i$ leads and $j$ 
    led the previous round is simply $\E[S^{j, i}]$. Thus, the expected duration is 
    $\frac{1}{n^2}\sum_{i, j} \E[S^{j, i}]$. Dividing the expected reward by the expected 
    round duration yields the desired time-averaged utility.
\end{proof}

\ifdeferproofs
\subsection{Coalition resistance} \label{sec:coalition}
\else 
\section{Coalition resistance} \label{sec:coalition}
\fi
In this section, we argue that $\sigma^*$, the honest-voting and early-proposing equilibrium, will be resistant even to coalition deviations.
Our results depend on the size of the coalition, which we characterize below with $k = m + c -n$:

\begin{table}[H]
  \centering
  \captionsetup{position=bottom,skip=8pt}
  \setlength{\tabcolsep}{12pt}
  \renewcommand{\arraystretch}{1.3}
  \begin{tabular}{|l|c|c|c|}
    \hline
    \textbf{Coalitions} & \textbf{Small} & \textbf{Medium} & \textbf{Large} \\
    \hline
    \hline
    Coalition size & $< k$ & $[k, m)$ & $\geq m$ \\
    \hline
    Example ($c = m = 2n/3$) & $< n/3$ & $[n/3, 2n/3)$ & $\geq 2n/3$ \\
    \hline
  \end{tabular}
  \caption{Coalition size categories, where $k = m + c -n$.}
  \label{tab:coalition-sizes}
\end{table}

Suppose that coalitions vote $0$ for the duration of coalition leaders' proposals.\footnote{We prove that this is optimal for small coalitions, and assume it for large coalitions due to its intuitive appeal and because large coalitions may grow from small coalitions.} 
Consider a coalition leader that proposes with delay $\Delta$ in some round. For this proposal, the $m\ts{th}$ smallest vote comes from an honest, quorum-participating voter (and thus the vote is at least $\Delta$). For a large coalition, the aggregate always comes from a coalition member. Medium coalitions are less predictable. The aggregate will always come from an honest voter, but they may not have participated in the quorum (in fact, their vote could be arbitrarily small). 
\ifdeferproofs
\begin{toappendix}
    \ifdeferproofs
\section{Coalition Resistance Results (\cref{sec:coalition})}
\label{appendix:coalition_resistance}
\else
\fi
Here, we flesh out two of the arguments against coalition resistance that we sketched. To review, small coalitions have less than $k = m + c-n$ members and large coalitions have more than $m$ members. We show the following.

\begin{enumerate}
    \item For small whales and rational coalitions, the optimal strategy is to always vote $0$ on each other's proposals, and propose with $\Delta = 0$. 
    \item For large coalitions that vote $0$ on each other's proposals, if the average latencies are low relative to the timeout, then large coalitions would not form. 
\end{enumerate}

\subsection{Small coalitions}
Notice that the reported duration may be less than $\Delta$ if the coalition is too large. 
For $m = c = 2n/3$, suppose that $n/3$ validators got the block so late that they receive the next 
block before the current one (since $2n/3$ is needed for the quorum, the chain can advance without these validators). 
Then, if $n/3$ reports 0, the reported duration is going to be 0 as well, even though the leader took non-zero duration 
to propose the block. Formally, we have the following result.
\begin{lemmarep}\label{lm:small}
    Let $k = m + c -n$. If at least $n-k$ validators vote honestly, then for any round with delay $\Delta$, the $m\ts{th}$ smallest vote will be at least $\Delta$. If fewer than $n-k$ validators vote honestly, then the $m\ts{th}$ smallest vote may be less than $\Delta$.
\end{lemmarep}
\begin{proof}
    We count the leader's delay $\Delta$ from the earliest time they could have proposed, which is after they receive the quorum. At that point, at least $c$ validators have received the previous block, and will therefore observe an duration of at least $\Delta$. Since $k$ of these validators may dishonest, only $c-k = n-m$ are guaranteed to report a duration of at least $\Delta$. Thus, the $(n-m)\ts{th}$ largest vote (which is the $m\ts{th}$ smallest vote) is at least $\Delta$.

    If fewer validators voted honestly, then note that all non-quorum validators may receive the previous block after the current leader's block (and would thus observe a negative round time), and dishonest validators can similarly report a negative round time. In this case, the $m\ts{th}$ smallest vote would be less than $\Delta$ (in fact, it would be negative).
\end{proof}
This lemma is important because in small coalitions, delaying is always punished by an honest voter, while in larger coalitions that may not be the case. Now, we turn to whales, and show that small whales do not benefit from deviating, as they are already co-located. 
\begin{propositionrep} \label{thm:small-whale}
    Consider a whale of size at most $k$. The optimal strategy for the whale is the early-proposing honest strategy, $\sigma^*$.
\end{propositionrep}
\begin{proof}
    Recall that the time-averaged utility under $\sigma^*$ for node $i$ is:
    \begin{align*}
         \frac{\sum_j\E[M(S^{j, i})] + B\left(\E\left[S^{j, i}\right]\right) - b\E\left[E^{i-j}_{m}\right]}{\sum_{i, j} \E[S^{j, i}]}
    \end{align*}
    Notice that the denominator (the duration) cannot decrease, as in $\sigma^*$ everyone proposes as early as possible. So, the only way for the whale to benefit is to increase the numerator, which is their reward. We split into cases, based on whether the whale deviates by delaying their proposal. First, suppose that they don't delay, and propose at $\Delta = 0$. Then, they could only benefit by lowering their latency penalty ($\sum_j b\E[E^{i-j}_{m}]$) through dishonest voting. However, since a whale has no latency to its own nodes, the whale already votes $0$ for their rounds under the honest strategy. So, the whale can't lower their latency penalty.

    Now, suppose the whale proposed at $\Delta > 0$. From \autoref{lm:small}, we know that the $m\ts{th}$ smallest vote will be at least $\Delta$. Since $\langle M, B \rangle$ are time-decreasing, a proposal of $\Delta = 0$ would maximize the whale's utility. 
\end{proof}
We now consider rational coalitions, and show that while small rational coalitions can benefit from deviating (in their voting strategy), they still propose early. 
\begin{theoremrep} \label{thm:small-rational}
    Consider a rational coalition of size at most $k$. The optimal strategy for the coalition members is to always vote $0$ on each other's proposals (despite any latency they may observe), and propose at $\Delta = 0$.
\end{theoremrep}
\begin{proof}
    The previous proof illustrates why delaying one's proposal is suboptimal, no matter how the $k$ other coalition members vote. Now, given that everyone is still proposing at $\Delta = 0$, voting $0$ clearly reduces the latency penalty as much as possible.

    In particular, if everyone votes at $0$, the time-averaged utility of coalition member $i$ is almost the same as in the honest case, except for the latency penalty. That penalty was proportional to $\sum_{j}b\E[E^{i-j}_m]$ before, but now it is $\sum_j b\E[\tilde{E}^{i-j}_{m-k}]$, where $\tilde{E}^{i-j}_{m-k}$ is the $(m-k)\ts{th}$ order statistic of $\{L_{i \to k} - L_{j \to k}\}_{k \notin C}$, where $C$ is the coalition. In other words, the latency penalty now depends on the $(m-k)\ts{th}$ smallest latency among the $n-k$ honest players, instead of depending on the $m\ts{th}$ smallest latency among all honest players. In particular, this latency penalty in the coalition can never be larger than the latency penalty in the honest strategy. Further, from $i$'s perspective, the most valuable coalition is the one which features the nodes furthest from $i$ (i.e., those nodes which have highest latency to $i$).  
\end{proof}

\subsection{0-voting large coalitions}
What if everyone in the coalition had $m$ free votes? In other words, suppose the coalition is large ($|C| \ge m$), and they vote $0$ for coalition leaders.\footnote{We consider the restricted set of large coalitions that always vote $0$ for each other because this voting strategy is optimal for small coalitions, and thus may arise naturally if a coalition starts small and grows. There is also an incentive for the coalition to grow, as members' utilities rise as the size approaches $k$.} Is this large coalition profitable?

To rule this out, we identify when the timing game behaves similarly to a prisoner's dilemma. Consider two different strategy profiles: one in which everyone proposes early and one in which everyone proposes late but votes $0$. Our approach, which uses voting to ensure that the validators can all commit to proposing early, is robust to large coalitions only if the validators would prefer the early-proposal equilibrium to the late-proposal equilibrium. In terms of the prisoners dilemma, validators preferring the late-proposal equilibrium is equivalent to agents preferring the ``all defect'' equilibrium. 

\autoref{thm:large-coalition} characterizes when enough validators prefer the early-proposing equilibrium to make any large, late-proposing coalition infeasible. This is because enough validators would want to leave that the large coalition would become small (and thus their optimal strategy would be early-proposing). This matches a prisoner's dilemma; each agent would would individually prefer to propose late (defect) but, many agents would rather everyone propose early (cooperate) than late.

To introduce the theorem, we first define a ``low-latency'' setting. To do so, let $\E[\bar{S}^i] = \frac{1}{n}\sum_j \E[S^{j, i}]$ denote the expected starting time for validator $i$ and let $\E[\overline{E}^{i}_x] = \frac{1}{n}\sum_j \E[E^{i-j}_x]$ denote the expected $x^{th}$ smallest latency differential for validator $i$. 
\begin{definition}[$z$-low latency]
    For $z \in [0, 1]$, a set of validators and their latency distributions satisfies $z$-low-latency if there exist at most $k$ nodes $i$ such that $\E[\overline{S}^i +\overline{E}^{i}_{\min(c, m)}] > z\tau$.
\end{definition}

Note that $\E[\overline{S}^i +\overline{E}^{i}_{\min(c, m)}]$ is the average reported round duration when validator $i$ leads. Thus, the condition requires that $n-k$ validators are well-connected, in that their average reported round duration is a $z$ fraction of the timeout. 

\begin{theoremrep} \label{thm:large-coalition}
    Suppose a large coalition $C$ votes $0$ for leaders from the coalition and votes honestly for other leaders, while the rest of the validators use truthful strategy $\sigma^*$.
    If the system has $z$-low-latency with $z \le \frac{n-c}{3n-c}$, 
    then $|C| - k$ validators within the coalition would prefer defecting to the early-proposing truthful strategy.
\end{theoremrep}
\begin{proof}
    We start by bounding the utility of coalition members.
    Observe that if a coalition leader $i$ delayed by $\Delta$, 
    their reward for that round would be $M(S^{j,i} + \Delta) + b_0$.        
    Due to the optimal-delaying oracle, coalition leaders will know the time $\Delta^*$ 
    that is the latest time they can propose and not timeout, and their utility is clearly maximized at $\Delta^*$. 

    Now, note that $\Delta^*$ averaged over the previous leader can be derived from 
    $\overline{S}^i + \overline{E}_c^i + \Delta^* = \tau$
    where $\overline{E}^i_c = \frac{1}{n} \sum_{k \in [n]} E^{i - k}_c$.
    This means that 
    the $c\ts{th}$ slowest validator sees $i$'s block exactly when 
    their local clock is $\tau$. Thus, $\E[\Delta^* + S^i] = \tau 
    - \E[\overline{E}_c^i]$, which is the average round duration of coalition leader $i$. This means that their utility per-round is $M(\tau - 
    \E[\overline{E}_c^i]) + b_0 \le M(\tau) + b_0$. Thus, the time-averaged utility 
    of coalition member $i$ is at most:
    \begin{align*}
        \frac{\frac{1}{n}(M(\tau) + b_0)}{\frac{|C|}{n}\sum_{i \in C} \tau - \E[\overline{E}_c^i] + \frac{n - |C|}{n}\sum_{i \notin C} \E[\overline{S}^i]}
    \end{align*}
    We now bound the increase in utility. 
    \begin{claim}
        The average round duration is at least $\frac{|C|-k}{zn}$ times higher for the coalition than in the honest equilibrium.
    \end{claim}
    \begin{nestedproof}
        From $z$-low-latency, we know that at least $|C|-k$ validators in the coalition have $\E[S^i] \le z\tau - \E[\overline{E}_c^i]$. These validators will increase their round duration from at most $z\tau - \E[\overline{E}_c^i]$ to $\tau - \E[\overline{E}_c^i]$, which is an multiplicative increase of more than $z$. Since no validator takes less time under the coalition, the total round duration increases by a multiplicative factor of at least $\frac{|C|-k}{nz}$.
    \end{nestedproof}    
    \begin{claim}
        For any coalition $C$, there exists at least $|C|-k$ members whose reward when leading a round in the coalition is at most $\frac{2}{1-z}$ times as high as in the honest equilibrium.
    \end{claim}
    \begin{nestedproof}
        In the honest equilibrium, the reward for validator $i$ when leading is $\E[M(S^i)] + B(\E[S^i]+\E[\overline{E}^i_m])$, which is highest for nodes with lowest $\E[S^i] + \E[\overline{E}^i_m]$. From $z$-low-latency, we know that at least $|C|-k$ validators in the coalition have $\E[S^i+\overline{E}^{i}_m] \le z\tau$. For these validators, their reward in the honest equilibrium is at least $B(\E[S^i]+\E[\overline{E}^i_m]) \ge B(z\tau)$. Since $B$ is linear and $B(\tau) = 0$, $B(z\tau) = (1-z)b_0$.

        In the coalition, all validators get reward $b_0 + M(\tau)$. Since $\langle M, B \rangle$ is time-decreasing, we know that $M(\tau) + B(\tau) \le M(0) + b_0$. Letting $B(\tau) = M(0) = 0$, we get $M(\tau) \le b_0$. Thus, the validators get reward at most $2b_0$ in the coalition. Dividing these terms results in a multiplicative gain of at most $\frac{2}{1-z}$ for $|C|-k$ validators in the coalition. 
    \end{nestedproof}
    Putting the two claims together, we get that at least $|C|-k$ members in the coalition would prefer to leave the coalition if $\frac{|C|-k}{nz} \ge \frac{2}{1-z}$, which simplifies to $z \le \frac{|C|-k}{2n+|C|-k}$. The right hand size is minimized when $|C| = m$. Since $m-k = n-c$, if $z \le \frac{n-c}{3n-c}$, $|C|-k$ members of the coalition prefer the honest equilibrium to the coalition.
\end{proof}
Under the conditions in the theorem, we see that 0-voting large coalitions that propose late become unsustainable. 
$|C|-k$ can ``defect'' from the coalition, and if they do so, then the coalition 
becomes small ($|C| \le k$), and thus will not compromise responsiveness (as well as facing practical barriers to adoption due to a multiplicity of coalitions).

To ensure that the conditions of the theorem hold in practice, there are multiple options.
Note that the ``low-latency'' condition requires that latencies are bounded away from $z\tau$. Thus, one can 
increase $\tau$ to make this condition hold. Intuitively, this makes the late-proposing equilibrium worse by decreasing the speed of the system. 
However, increasing $\tau$ would require higher block rewards. This is because $B(\tau)=0$, and to ensure $\langle M, B\rangle$ is time-decreasing, the slope of $B$ must not change; thus, $B(0)$ must increase.
 
Another option is to uniformly increase the block reward function (i.e., shift it up).
One reason that validators may prefer the late-proposing equilibrium is that they always get rewarded as if their rounds had zero duration, as opposed to the reward of $B(\E[S^i]+\E[\overline{E}^i_m])$ in the honest equilibrium. Increasing the block reward function will thus decreases the multiplicative advantage that occurs from gaining $b_0$ in the late-proposing equilibrium. 

\end{toappendix}
\else
\begin{toappendix}
\section{Coalition Resistance Results (\cref{sec:coalition})}
\label{appendix:coalition_resistance}
\end{toappendix}
\fi
We have three arguments for coalition resistance. 
\begin{enumerate}[(A)]
    \item
    Small coalitions will never compromise responsiveness:
    \begin{theorem}\label{thm:small-informal}
        For small whales and rational coalitions, the optimal strategy is to always vote $0$ on each other's proposals, and propose at $\Delta = 0$.
    \end{theorem} 
    \item We argue \textit{informally} that small and medium rational coalitions are unstable. Suppose multiple, overlapping coalitions could form, as is true in practice. If validator $i$ acquires $m$ ``free'' votes, then they would propose late, as they are guaranteed the maximum block reward regardless. Further, other validators may want to enter into a $0$-voting pact with $i$, as doing so will increase their utility by decreasing their own latency penalty. They would only reject this pact if they believed that $i$ would have $m$ free votes, and thus delay. Estimating this probability would be challenging in practice, rendering small and medium coalition formation highly unpredictable. Taken together, we believe that small or medium rational coalitions will not form in practice, as coalition members will have an incentive to join additional coalitions.

    Note that this argument assumes that rational coalitions can neither (i) detect and penalize coalition-defectors, nor (ii) identify surreptitious participation in other coalitions.
    In practice, small coalitions could detect if one of their members proposes much later than they should (due to their secret membership in other coalitions). However, that member can instead make a smaller delay which could be difficult to distinguish from latency fluctuations. This delay would still lead to them benefiting at the cost of the rest of the network, and thus the same qualitative issues arise.  
    Coalition capabilities and formation dynamics are complex~\cite{kelkar2025breaking,RAY2015239,ray2007game,shenoy1979coalition}; we defer the meta-game of coalition formation to future work.
    \item Suppose that small rational coalitions do grow into large coalitions. Since the optimal voting strategy for small coalitions is to vote $0$ on coalition leaders, we consider large coalitions that follow this same strategy. We also introduce a condition called $z$-low-latency, which (roughly) requires that at most $k$ validators receive aggregated timeliness votes greater than $z\tau$ in the honest equilibrium (where $z\in [0, 1]$). We believe that this is a natural condition, as it would be satisfied if the timeout is set high enough to account for large latency fluctuations. Using this, we show the following theorem. 
    \begin{theorem}
        If the system has $z$-low-latency with $z \le \frac{n-c}{3n-c}$, 
        then $|C| - k$ validators within the coalition would prefer defecting to the early-proposing truthful strategy.
    \end{theorem}
    Thus, large, low latency, $0$-voting coalitions will not form, because a large fraction of the coalition would prefer to return to the honest equilibrium.
\end{enumerate}
\ifdeferproofs
\else
We defer full results for arguments (A) and (C) to \cref{appendix:coalition_resistance}.
\fi

Our coalition resistance results are summarized in \autoref{tab:coalition-args}. While we do not have arguments against medium or large whales, such a whale must control a very large fraction of the total stake, and is thus unlikely to exist or violate the assumptions of the underlying BFT protocol itself.
\begin{table}[H]
  \centering
  \captionsetup{position=bottom,skip=8pt}
  \setlength{\tabcolsep}{10pt}
  \renewcommand{\arraystretch}{1.2}
  \begin{tabular}{@{} l cc @{}}
    \toprule
    & \textbf{Whale} & \textbf{Rational} \\
    \midrule
    \textbf{Small}  & (A)        & (A), (B) \\
    \textbf{Medium} & \textemdash & (B)      \\
    \textbf{Large}  & \textemdash & (C)      \\
    \bottomrule
  \end{tabular}
  \caption{Cells show which arguments apply to which sizes and types of coalitions.}
  \label{tab:coalition-args}
\end{table}
\ifdeferproofs
\else

\fi

\section{Fairness} \label{sec:fairness}
Prior works (e.g., \cite{PS17}) defined a blockchain as ``fair'' (among validators) when the total welfare is distributed proportionally to stake. 
One of the main relevant obstacles to fairness is latency---with time-varying block reward, the utility of a validator decreases as it experiences higher latency to other validators.
Note that latency-based unfairness already exists with static block rewards.
Validators with lower latency can propose closer to the timeout while still getting their blocks approved, resulting in more MEV than validators with higher latency. Throughout our analysis of fairness, we assume that validators are in the early-proposing and honest-voting equilibrium with dynamic block rewards, and in the late-proposing equilibrium with static block rewards.

We will capture \emph{unfairness} as the utility ratio between the player with highest utility and the lowest utility.
At best, this ratio is 1 and the equilibrium is completely fair. As we will see in this section, the ratio is often less than 1.
\begin{definition}[latency advantage]
    Let $U_i(\vec{\sigma})$ be the utility of player $i$ at some fixed strategy profile $\vec{\sigma}$.
    \[
    \textsc{Advantage}(\vec{\sigma}) 
    = \frac{\max_i U_i(\vec{\sigma})}
    {\min_i U_i(\vec{\sigma})} - 1
    \]
\end{definition}

\subsubsection{Closed-form results.}
We first analyze fairness in two simple network topologies with deterministic latencies to gain intuition about fairness under dynamic block rewards.
In the \emph{line model}, $n$ validators are spaced evenly on a line, and the latency between two validators is proportional to their distance.
In the \emph{cluster model}, $n$ validators are placed into two clusters, and the latency between two validators is low if they are in the same cluster and high otherwise. 
The following informal theorem summarizes results from 
\cref{thm:fairness-line,thm:fairness-cluster,thm:fairness-late}.
\begin{theorem}[informal]
In the line model and the cluster model, with deterministic latencies and
reasonable parameters, fairness is worse under dynamic block rewards than static block rewards.
\end{theorem}
We defer full results to \cref{apx:fairness}.
The theorem suggests that shifting the proposal times earlier has a fairness cost, although the cost remains small in the scenarios we analyze. For instance, if we plug in reasonable parameters inspired by current Ethereum settings, the unfairness (i.e., the percentage increase between the lowest utility player and the highest utility player) remains within a couple percentage points of the static reward setting, in both the line and cluster models (see \cref{fig:fairness-plots}).

\subsubsection{Simulation details.}
We also simulate both the dynamic and static block reward settings in network topologies inspired by real-world latencies.
On a world latency model, where pairwise latencies are sampled from a lognormal, we observe that the unfairness under dynamic rewards is only slightly worse than under static rewards (see \cref{fig:fairness-world}).
We defer full details to \cref{appendix:sim}.
\begin{figure}[h]
    \centering
    \includegraphics[scale=0.75]{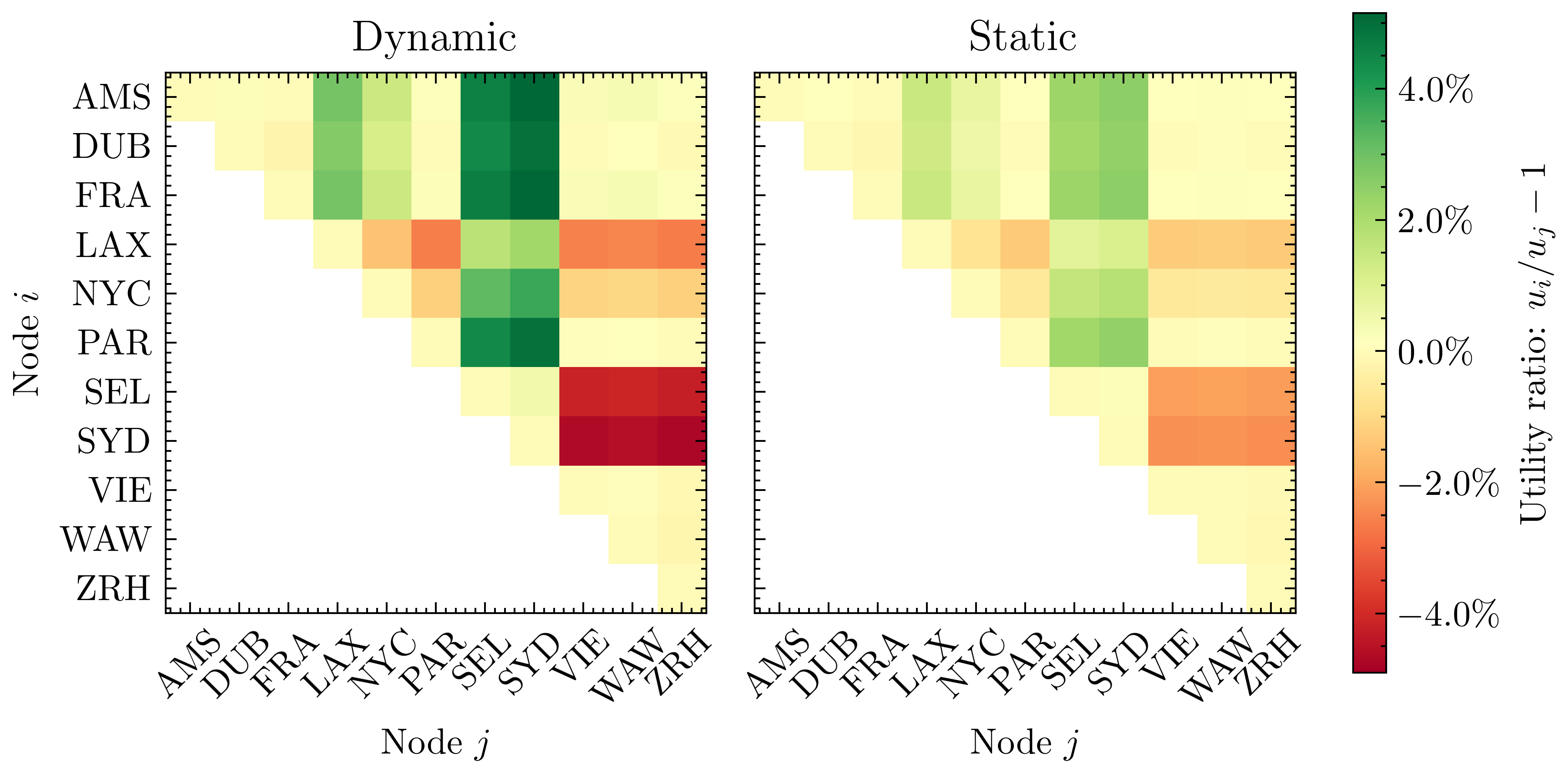}
    \caption{Simulation of 100 nodes placed in cities around the world, where
    the latency between $i$ and $j$ is log-normal with mean as average ping
    latencies. 
    }
    \label{fig:fairness-world}
\end{figure}

\section{Guidance for practice} \label{sec:reward}
We now discuss some details that may be relevant for implementing our work. 

\subsubsection{Empirical feasibility of time-decreasing rewards.}
Due to our requirements on total reward, observe that the maximum block reward needs to be greater than the accumulated MEV at timeout.
\begin{observation} \label{lem:block-reward-needs-to-be-high}
    For MEV $M(x) = \mu x + \mu_0$ and block reward $B(x) = b_0 - bx$ where $B(\tau) = 0$,
    $\langle M , B \rangle$ is time-decreasing (\cref{def:time-decreasing}), if and only if 
    $b_0 > \mu \tau$.
\end{observation}
We take the example of Ethereum to illustrate that $b_0$ can indeed be set sufficiently high in practice. 
Prior empirical studies on Ethereum have looked at the marginal value of time in the block building auction\footnote{On 
Ethereum, most validators auction off the right to build their block to sophisticated block builders that are able to 
extract more MEV. Assuming that the auction is competitive, this provides a close lower bound for estimating MEV.}.
The increase in MEV is estimated as $\mu = 6.5 \times 10^-6~\mathtt{ETH}/\mathtt{ms}$ in \cite{SSTPSM23} 
and $\mu = 5.71 \times 10^-6~\mathtt{ETH}/\mathtt{ms}$ in \cite{OKVKMT23}.
For a block reward curve of the form $b_0 - \frac{b_0}{12~\mathtt{s}}x$, 
$b_0$ needs to be greater than $\mu \cdot 12000~\mathtt{ms}$ per Observation~\ref{lem:block-reward-needs-to-be-high}.
For $\mu = 6.5 \times 10^-6~\mathtt{ETH}/\mathtt{ms}$, this boils down to $b_0 > 0.078~\mathtt{ETH}$.
Fortunately, the total consensus reward per block in Ethereum is much greater than $0.078~\mathtt{ETH}$~\cite{E23,OKVKMT23}.

\subsubsection{Setting $m$.}
Our results only require that $m > n-c$. This is a natural condition, as consensus protocols assume that at most $n-c$ nodes can be faulty, implying that the timeliness vote should not come from a faulty node. Moreover, since $c>n/2$ for all reasonable BFT protocols, this implies that $m$ is not required to be higher than the quorum size $c$ of the BFT protocol. 
The coalition size condition is also quite natural.
In general, $m$ should be set sufficiently high for resistance to coalitions (including whales), but not so high that correct nodes in BFT protocols are insufficient to form the timeliness vote. If fewer than $m$ validators submit a timing vote, the block reward would be zero, erasing any incentive to propose early. 
One natural setting is $m = 2(n-c)$, which ensures that $k = n-c$, and so the small coalition size accounts for all byzantine faulty nodes in BFT consensus models.
Concretely, if $c = 2/3n$, then a natural setting is $m = 2/3n$, resulting in coalitions of all $n/3$ byzantine nodes to be characterized as small and not affecting the protocol's responsiveness (Theorem~\ref{thm:small-informal}). 

\subsubsection{Target duration.}
While responsiveness is desirable, in practice, decentralized protocols require that the round times are not too fast and slow nodes are able to keep up.
Observe that the fastest $c$ nodes can drive the protocol, leaving the rest behind.
In practice, blockchain protocols attempt to have a \textit{target duration} (minimum round time), e.g., 0.3 seconds in Solana~\cite{solana-chorusone-timing}.
Our block reward function can be simply modified as follows to ensure that proposing before the target duration is never beneficial: the block reward starts at $0$ for time 0, and increases until the target duration, and then decreases from there. This ensures that leaders will not want to propose before the target duration, but also not want to delay their proposal beyond the target duration.

\section{Conclusion} \label{sec:discussion}
We modeled the timing game in responsive BFT protocols and argued for 
measuring proposal timeliness via a vote. Our results indicate that 
with a carefully chosen block reward curve and timeout, responsiveness 
itself can be used to align timing incentives.   

For future work, there are several avenues to explore. 
The timeliness measurement requires a voting protocol that 
needs to be designed and implemented. 
The design space for this protocol is rich. One option is to 
make the next leader aggregate timeliness votes and post
(potentially with a delay). While this would be easy to implement, 
it increases the scope of strategic manipulations by allowing the aggregator 
to filter. Another promising approach is to aggregate and post timeliness votes on-chain.

Alternative incentive knobs to make early-proposing desirable should also be 
explored. Our discussion in \cref{apx:leader} could serve as a first step in 
considering more advanced leader-election schemes that punish late-proposing.
Moreover, as these approaches are not necessarily contradictory, further analysis 
of hybrid approaches could be explored.

Our coalition resistance results could also be extended and formalized further. For example, one could define a meta-game, where validators can join or leave coalitions over time. This would allow for formalizing our argument against small and medium rational coalitions and investigating further coalition dynamics.

\section*{Acknowledgments}
We would like to thank Michael Setrin and Matthew Weinberg for useful comments on earlier drafts.

\begin{toappendix}
    \section{Fairness Results (\cref{sec:fairness})} \label{apx:fairness}
In this section, we analyze fairness for two concrete latency models under dynamic or static block rewards. We refer to the former as the ``early-proposing equilibrium'' and the latter as the ``late-proposing equilibrium''.

We analyze two latency models inspired by concerns of geographical collocation. The first model is one where nodes are relatively evenly dispersed, and the second is one with two clusters of nodes, with $2/3$ and $1/3$ of the nodes in the two clusters.
Our results indicate that fairness slightly degrades when moving to early-proposing equilibrium.
At the end of the section, we provide an empirical analysis of latency and fairness experienced 
in a real-life distributed system and explore broadcast protocols as a promising avenue for mitigating unfairness.
Our theoretical results are derived via deterministic latencies and consider the following two models:
\begin{itemize}
    \item \emph{Line model:} Nodes are evenly dispersed on a number line. 
    Latency is directly proportional to distance. 
    Intuitively, we want to measure the difference in utility between the center node and the endpoint node.
    \item \emph{Cluster model:} Nodes are clustered into two clusters. 
    Within a cluster, latency is very low and across the clusters latency is very high. 
    Intuitively, we will compare the gap between a node in one cluster and a node in the other cluster.
\end{itemize}
In \cref{appendix:sim} we also simulate these models with lognormal latency 
distributions and get similar results to the theoretical ones.

 \subsubsection{Reward slope and timeout slack.}
Since early-equilibria requires the slope of block reward to 
be greater than the MEV slope (Observation~\ref{lem:block-reward-needs-to-be-high}), 
we will assume that $b > \mu$ throughout this section. Therefore, we will phrase our results in 
terms of $k$ such that $b = k\mu$ for $k > 1$. 
Another intuitive parameter we use in our results is the timeout slack, 
how many times larger the timeout is as compared to the slowest node's start time in order to allow for everyone to participate.
Intuitively, when consensus parameters are set, the unhappy path timeout gets set based on how early 
the slowest node can propose their blocks. 
Let $d$ represent this slack in the timeout. 

\subsubsection{Late equilibria.}
Here, we assume a fixed block reward of $b_0$. We assume validators will as late as possible while ensuring that every block still makes it in.
Hence we propose at $\Delta^* = \tau - S^{j,i} - E^{i - j}_c$ where $i$ is the current leader preceded by leader $j$.
Note that this is even better than proposing late in real life since there is no risk of missing blocks. 
We call this the late proposing strategy profile $\vec{\sigma_\text{late}}$.

Throughout the section, we theoretically analyze deterministic latency models, where links have fixed latency both on the line and 
the cluster model. Later on, we simulate both models with log-normal latency distributions.\footnote{For all our results, we assume that $n$ is divisible by 2 and 3 to simplify the expressions.}

We observe that there is a cost to early equilibria in terms of fairness.
For parameters\footnote{For example, for high performance chains, 
500~\texttt{ms} block times with a 2500~\texttt{ms} timeout and a quorum threshold at $2/3$ is very reasonable. 
The parameter $k$ depends on calibration of rewards. 
For a discussion of reward parameters see \cref{sec:reward}.} 
such that $k = 1.5$, and $m = c = 2/3$, we see that 
the fairness to utilizing early-equilibria compared to late (see \cref{fig:fairness-plots}) is up to $8\%$ worse in terms of advantage.

\subsection{Nodes on a line} \label{appendix:fairness1} 
We begin by considering deterministic latency. 
Validators are evenly spaced along a number line, with validator $i$ at number $i$. 
The latency between $i$ and $j$ is just $|j-i|$. We aim to measure the fairness index,
which occurs between validator $n/2$ (the most centrally located) and validator $0$. 
To do so, we assume that block reward and MEV are linear: 
$B(x) =b_0-bx$ and $M(x) = \mu x + \mu_0$. 
We also assume that $m = c = 2n/3$. 
We call this the \emph{deterministic line} model.

We now compute utility ratios to measure fairness in the deterministic line model.
First observe that the ratio of their utilities in the honest equilibrium 
depends on $\overline{S}^i$ and $\overline{E}^{i}_m$ for $i \in \{0,n/2\}$.
\begin{lemma} For the early-proposing strategy profile $\sigma^*$, if 
$b' > m$, then the fairness index in the deterministic 
line model is given by the following:
\[
\textsc{Advantage}(\vec{\sigma^*})
= \frac{(b_0 + \mu_0)-(b-\mu)\E[\overline{S}^{n/2}] -b\E[\overline{E}_m^{n/2}]}{(b_0 + \mu_0)-(b-\mu)\E[\overline{S}^{0}] - b\E[\overline{E}_m^{0}]} - 1
\]
\end{lemma}
\begin{proof}
    Observe that the time contribution of average round duration cancels out in the ratio of two utilities. 
    Hence, consider the reward earned in expectation by each player. 
    Observe that player $i$ becomes leader right after all players $j$ with probability $1/n^2$, 
    in which case it proposes at $S^{j,i}$. Hence, using $\overline{S}^i = \frac{1}{n}\sum_{j \in [n]} S^{j,i}$ and 
    $\overline{E}^{i} = \frac{1}{n} \sum_{j \in [n]} E^{i-j}_m$, we have the following:
    \begin{align*}
        \frac{U_{i}(\sigma^*)}{U_j(\sigma^*)} 
        &= \frac{b_0 - b(\E[\overline{S}^i + \overline{E}^{i}_m]) + \mu(\E[\overline{S}^i]) + \mu_0}{b_0 - b(\E[\overline{S}^j + \overline{E}^{j}_m]) + \mu(\E[\overline{S}^j]) + \mu_0}
    \end{align*}
    Since the player on the edge of the line and the middle of the line have the two extreme latency distributions 
    and $b' > m$, player $n/2$ earns the most and player $0$ earns the least. 
    We get the desired conclusion after applying linearity of expectation.
\end{proof}
Now, the expectations range only over the leader election randomness, as latencies are otherwise deterministic. 
\begin{lemma} \label{lem:det-quorum-line}
In the deterministic line model,
assume without loss of generality that $i \leq j$ for $i,j \in [n]$.
Then, 
\[
    \E[Q_c^{i \to j}] =
    \E[Q_c^{j \to i}] = \begin{cases}
        j-i & j-i+1\ge 2n/3 \\
        2n/3 - 1 + \kappa_{ij} 
        & j-i+1 < 2n/3 \land j-i+1+2w_{ij} \ge 2n/3 \\
        4n/3 -2w_{ij}- j + i - 2 & \text{otherwise} \\
    \end{cases} 
\]
where $w_{ij} = \min(n-j-1, i)$ 
and  $\kappa_{ij} = \mathbb{I}[2n/3 - (j-i+1) \mod 2 = 1]$.
\end{lemma}
\begin{proof}
We need to compute $Q_c^{j\to i}$ for all $j$ and for $i = 1, n/2$. 
So, we need the $(2/3)\ts{th}$ order statistic of $\{L_{j \to k} + L_{k \to i}\}_k=\{|j-k|+|k-i|\}_k$. 
Since each path is symmetric, $\E[Q_c^{j \to i}] = \E[Q_c^{i \to j}]$ 
so without loss of generality assume that $j \ge i$. Then, there are three cases to consider for $k$:
\begin{enumerate}
	\item For all $k \in [i, j]$, $L_{j \to k} + L_{k \to i} = j-i$.
	\item For all $k < i$: $L_{j \to k} + L_{k \to i} = j+i - 2k >j-i$.
	\item For all $k > j$, $L_{j \to k} + L_{k \to i} = 2k - (j+i) > j-i$.
\end{enumerate}
For the $(2/3)\ts{rd}$ fastest node $k$, 
note first that the nodes in $[i, j]$ are fastest. 
After those nodes are exhausted, the nodes $i-1$ and $j+1$ have the same latency: $j-i+2$. 
So, then we pair up these nodes while they exist. 
After that point (i.e., we hit $k = 1$ or $n$ on the sides), 
we select single nodes on the remaining side to get to the $2/3$ fastest node. 
Thus, $Q_c^{j \to i}$ can be broken into three cases:
	\begin{enumerate}[(a)]
		\item If $j-i+1\ge 2n/3$, then all $2n/3$ nodes are contained in $[i,j]$ so $Q_c^{j \to i} = j-i$.
		\item If $j-i+1 < 2n/3$ but $j-i+1+2\min(n-j-1, i) \ge 2n/3$, 
        then we can pair nodes from both sides of the line to form $2n/3$ votes so 
        $Q_c^{j \to i} = j-i+ 2 \left\lceil \left( \frac{2n}{3} - (j-i+1) \right) / 2 \right\rceil$. 
        The second term is the offset from the paired nodes before $i$ and after $j$. 
        This simplifies to $2n/3 - 1 + \kappa_{ij}$.
		\item Otherwise, $Q_c^{j \to i} = (j-i)+ 2\min(n-j-1,i)
        + 2\left( \frac{2n}{3}-(2\min(n-j-1, i) + j-i+1) \right)$. 
        Here, the $\min(n-j-1, i)$ term comes from pairing nodes outside $[i, j]$, 
        and the last term comes from the remaining nodes needed on one side. 
        These terms are multiplied by two, 
        since the offset increases by 2 each time. 
        This simplifies to 
        $Q_c^{j \to i} = 4n/3 -2\min(n-j-1, i) - j + i -2$.
	\end{enumerate}
\end{proof}

\begin{lemma} \label{lemma:s}
In the deterministic line model,
$\E[\overline{S}^{n/2}] = \frac{25n}{36} - \frac{2}{3}$%
and $\E[\overline{S}^{0}] = \frac{17n}{18}-\frac{7}{6}$.
\end{lemma}
\begin{proof}
Recall that $S^{j, i} = \max(Q_c^{j \to i}, L_{j\to i})$. In this model, the quorum time always takes longer than the second term (due to these deterministic latencies obeying the triangle inequality). 
By \cref{lem:det-quorum-line}, we consider three cases for both $\overline{S}^0$ and $\overline{S}^{n/2}$.

For $\overline{S}^0$ observe that the middle case vanishes, as there are no nodes before the first. 
Further, $\min(n-j-1, i)=0$, as $i=0$. Putting this all together:
\begin{align*}
\overline{S}^0
&= \frac{1}{n}\sum_{j \in [n]} S^{j, 0} =  \frac{1}{n}\sum_{j \in [n]} Q^{j \to 0}_{2/3} \\
&= \frac{1}{n} \left(\sum_{j=0}^{2n/3-2} \left( \frac{4n}{3} -j - 2 \right) + \sum_{j = 2n/3-1}^{n-1} j \right)
& \text{\cref{lem:det-quorum-line}} \\
&= \frac{1}{n} \left( \frac{(n-1)(2n-3)}{18} + \frac{(n+3)(5n-6)}{18} \right) \\
&= \frac{17n}{18} - \frac{7}{6}
\end{align*}

Now, we consider $\overline{S}^{n/2}$. Note that our earlier derivations assumed that $j \ge i$. However, $S^{j, i} = S^{i, j}$ since $L_{j \to i} = L_{i \to j}$. So, we simply compute $S^{n/2, j}$ for $j < n/2$. Note that in this case, $j - n/2 + 1$ is always less than $2n/3$, so the first case for $Q_c^{j \to n/2}$ vanishes. Further, whether there's enough nodes for the ``pairing'' strategy to work simply depends on if $j \le n/6$ or $j > 5n/6$; otherwise, we use the second case. Putting it all together:
\begin{align*}
n\overline{S}^{n/2} 
&= \sum_{j \in [n]} S^{j, n/2} \\
&= \sum_{j=0}^{n/2-1} S^{n/2, j} + \sum_{j = n/2}^n S^{j, n/2}\\
&= \sum_{j=0}^{n/2-1} Q^{n/2 \to j}_{2/3} + \sum_{j = n/2}^{n-1} Q^{j \to n/2}_{2/3}
\end{align*}
Now, we compute these terms using \cref{lem:det-quorum-line}.
The four segments we consider are $[0, n/6 - 2], [n/6 - 1, n/2 - 1], [n/2, 5n/6-1], [5n/6,n-1]$.

First consider $[0, n/6 - 2], [n/6 - 1, n/2 - 1]$.
Observe that for $Q^{n/2 \to j}_{2/3}$, $\min(n/2 - 1, j) = j$. 
Moreover, for $j \in [0, n/6 - 2]$ case 3 applies and for $j \in [n/6 - 1, n/2 - 1]$ case 2 applies 
of \cref{lem:det-quorum-line} always applies.
Therefore,
\begin{align*}
    \sum_{j=0}^{n/2-1} Q^{n/2 \to j}_{2/3}
    &= \sum_{j=0}^{n/6-2} \left( \frac{4n}{3} -j + \frac{n}{2} - 2 \right) 
    + \sum_{j=n/6-1}^{n/2-1} \left( \frac{2n}{3} - 1 + \kappa_{ij} \right) \\
    &= \left( \sum_{j=0}^{n/6-2} \left( \frac{5n}{6} -j + \frac{n}{2} - 2 \right) \right)
    + \left( \frac{n}{3} + 1 \right) \left( \frac{2n}{3} - 1\right) + \frac{n}{6} \\
    &= \left(\frac{n^2}{8} - \frac{11n}{18} + 1 \right) + \left( \frac{2n^2}{9} + \frac{n}{2} - 1\right)
\end{align*}
where we use the fact that there are $n/6$ integers that make $\kappa_{(n/2)j} = 1$ 
in the second segment for $n$ divisible by 6.

Similarly, for $[n/2, 5n/6-1]$ case 2 applies and for $[5n/6,n-1]$ case 3.
We use the fact that in $[5n/6,n-1]$, $\min(n - 1 - j, n/2) = n-1-j$ and in $[n/2, 5n/6-1]$
there are $n/6$ that make $\kappa_{(n/2)j} = 1$.

\begin{align*}
    \sum_{j = n/2}^{n-1} Q^{j \to n/2}_{2/3}
    &= \sum_{j=n/2}^{5n/6-1} \left( \frac{2n}{3} - 1 + \kappa_{ij} \right)
       + \sum_{j=5n/6}^{n-1} \left( j - \frac{n}{6} \right) \\[6pt]
    &= \left( \frac{n}{3} \right) \left( \frac{2n}{3} - 1 \right) + \frac{n}{6}
       + \left( \sum_{j=5n/6}^{n-1} \left( j - \frac{n}{6} \right) \right) \\
    &= \left( \frac{2n^2}{9} - \frac{n}{6} \right) + \left( \frac{n^2}{8} - \frac{n}{12} \right)
\end{align*}
Lastly, adding all and dividing by $n$, we get the following:
\[
    \overline{S}^{n/2}
    = \frac{1}{n} \left( \sum_{j=0}^{n/2-1} Q^{n/2 \to j}_{2/3} + \sum_{j = n/2}^{n-1} Q^{j \to n/2}_{2/3} \right) 
    = \frac{25n}{36} - \frac{2}{3}
\]

\end{proof}

\begin{lemma} \label{lemma:q}
$\E[\overline{E}_m^{0}] = \frac{7n-15}{18}$ and $\E[\overline{E}_m^{n/2}] = \frac{n}{36} - \frac{2}{3}$.
\end{lemma}
\begin{proof}
Recall that $E^{i-j}_m = f(m, \{L_{i \to k} - L_{j \to k}\}_k) = f(m, \{|i-k|-|j-k|\})$, i.e., the $m$th smallest number in the set. Now, suppose that $i < j$. Then, there are three cases:
\begin{enumerate}
	\item If $k < i$, then $i-k-(k-j) = i-j$.
	\item If $k \in [i, j]$, then $k-i -j +k = 2k-(i+j)$.
	\item If $k > j$, then $k-i-k+j = j-i$.
\end{enumerate}
Note that the values goes up monotonically with $k$. So, $E^{i-j}_m = L_{i\to m}-L_{j \to m}$, for $i < j$. 
A similar argument shows that when $j > i$, the values go down monotonically with $k$, 
and thus $E^{i-j}_m = L_{i \to (n-m)} - L_{j \to (n-m)}$. 
Note that $E_m^{i-i}=0$ and the $m\ts{th}$ index when indexing from 0 is $2n/3-1$.

We first compute $\overline{E}^0_m$:
\begin{align*}
n\overline{E}^0_m &= \sum_{j = 0}^{n-1} E_m^{0-j} \\ 
&= \sum_{j=1}^{2n/3-1} \left( \left(\frac{2n}{3} - 1\right) - 
\left(\frac{2n}{3} - 1 -j\right)\right) 
+ \sum_{j=2n/3}^{n-1} \left(\left( \frac{2n}{3} - 1 \right) - \left(j - \frac{2n}{3} + 1\right) \right) \\
&= \sum_{j=1}^{2n/3-1} j + \sum_{j=2n/3}^{n-1} \left(\frac{4n}{3}-j-2\right) \\
&= \frac{2n^2 - 3n}{9} + \frac{n^2-3n}{6} \\
&= \frac{n(7n-15)}{18}
\end{align*}
Thus, $\overline{E}^1_{m} = \frac{7n-15}{18}$, as desired.

Now, we compute $\overline{E}^{n/2}_m$:
\begin{align*}
n\overline{E}_m^{n/2}&=\sum_{j=1}^{n/2-1} E_m^{n/2-j} + \sum_{j = n/2+1}^n E_m^{n/2-j} \\
&= \sum_{j=1}^{n/2-1} \left( \left|\frac{n}{2} - \frac{n}{3} \right| - \left| j - \frac{n}{3} \right| \right) 
+ \sum_{j = n/2+1}^n \left(  \left|\frac{n}{2} - \left(\frac{2n}{3} - 1\right) \right| - \left| j - \left(\frac{2n}{3} - 1 \right) \right| \right)\\
&= \frac{n^2}{36} - \frac{2n}{3}
\end{align*}
So, $\overline{E}_m^{n/2} = \frac{n}{36} - \frac{2}{3}$, as desired.
\end{proof}

We now combine these results to prove the following theorem.
\begin{theorem}[line,\;$\sigma^*$] \label{thm:fairness-line} 
    In the line model with deterministic latencies, the advantage is:
    \begin{align*}
    \textsc{Advantage}(\sigma^*) = 
        \frac{
            d
            +\frac{\mu_0}{kn\mu}
            -\frac{k-1}{k} 
                \left(\frac{25}{36}- \frac{2}{3n} \right) 
            - \frac{1}{36} 
            - \frac{2}{3n}
            }
            {
            d
            + \frac{\mu_0}{kn\mu}
            -\frac{k-1}{k}
                \left(\frac{17}{18}-\frac{7}{6n} \right) 
            - \frac{7}{18}
            -\frac{5}{6n}
            }
    \end{align*}
\end{theorem}
\begin{proof}
    Recall that the ratio of utilities is:
    \begin{align*}
        \frac{U_{n/2}}{U_0} = \frac{(b_0+\mu_0)-(b-\mu)\E[\overline{S}^{n/2}] - b\E[\overline{E}_m^{n/2}]}{(b_0+\mu_0)-(b-\mu)\E[\overline{S}^{1}] - b\E[\overline{E}_m^{1}]}
    \end{align*}
    Plugging in the values from \cref{lemma:s,lemma:q}:
    \begin{align*}
        \frac{U_{n/2}}{U_0} 
        = \frac{(b_0 + \mu_0) - (b- \mu) \left(\frac{25n}{36} 
        - \frac{2}{3}\right) 
        - b\left( \frac{n}{36} - \frac{2}{3} \right)}
        {(b_0 + \mu_0)
        - (b-\mu)\left(\frac{17n}{18}-\frac{7}{6} \right) 
        - b\left( \frac{7n}{18} - \frac{5}{6} \right)}
    \end{align*}
    Now, note that the longest start time, $S^{0, n-1} = n-1 \approx n$. Suppose the timeout is a multiple of the longest start time, so that $\tau = dn$. 

    Now, since $B(\tau) = 0$, we know that $b_0 - b\tau = 0$. Since $b > \mu$, let $b = k\mu$ for $k > 1$. Then, we get that $b_0 - k\mu\tau = 0$, or $b_0 = k\mu\tau = k\mu dn$. 
    Now, the utility ratio simplifies to:
    \begin{align*}
        \frac{U_{n/2}}{U_0} 
        &= \frac{
            k{\mu}dn 
            + \mu_0
            - (k-1){\mu} \left(\frac{25n}{36}- \frac{2}{3}\right)
            - k{\mu}\left(\frac{n}{36}-\frac{2}{3}\right)
            }
            {
            k{\mu}dn
            + \mu_0
            -(k-1){\mu}\left(\frac{17n}{18}-\frac{7}{6} \right) 
            - k{\mu}\left(\frac{7n}{18} - \frac{5}{6}\right)
            } \\
        &= \frac{
            d
            +\frac{\mu_0}{kn\mu}
            -\frac{k-1}{k} 
                \left(\frac{25}{36}- \frac{2}{3n} \right) 
            - \frac{1}{36} 
            - \frac{2}{3n}
            }
            {
            d
            + \frac{\mu_0}{kn\mu}
            -\frac{k-1}{k}
                \left(\frac{17}{18}-\frac{7}{6n} \right) 
            - \frac{7}{18}
            -\frac{5}{6n}
            }
    \end{align*}
    As $n$ increases, this quickly converges to:
    \begin{align*}
        \frac{U_{n/2}}{U_0} &= \frac{d-\frac{25}{36} \cdot \frac{k-1}{k}  - \frac{1}{36}}{d-\frac{17}{18}\cdot\frac{k-1}{k} - \frac{7}{18}}
    \end{align*}
\end{proof}

\subsection{Nodes in two clusters} \label{appendix:fairness2} 
We now consider a different latency model. Suppose the validators are clustered into two clusters: $X$ and $Y$, where $|X| \ge |Y|$. Latencies within a cluster are small: $\eps$. Latencies across the cluster are larger ($l > \eps$). 
We call this the \emph{two cluster model}.
What is the ratio of utilities of a node in one cluster to the other?

As before, we assume that latencies are deterministic, and that block reward and MEV are linear
and $c = m = 2n/3$. 

We compute the utility ratios in the two cluster model.
For notational convenience, 
let $\bX$ denote the relative fraction $\bX = \frac{|X|}{n}$
and $\bY$ denote $\bY = \frac{|Y|}{n}$.

\begin{lemma} \label{lem:s-cluster}
    In the two cluster model, for $x \in X$ and $y \in Y$ the following holds:
    \[
    \E[\overline{S}^y] = \bX(\eps + l) + 2\bY l 
    \qquad \text{ and } \qquad 
    \E[\overline{S}^x] = \begin{cases}
       2\bX \eps + \bY (\eps+l) & |X| \ge c \\
       2\bX l + \bY (\eps+l) & \text{otherwise}
    \end{cases}
    \]
\end{lemma}
\begin{proof}
    Consider $L_{j \to k} + L_{k \to i}$. There are three cases:
    \begin{enumerate}
    	\item If $j$ and $i$ are in different clusters, then this is $\eps + l$ (one hop within the cluster, the other hop between clusters).
    	\item If $j, i$, and $k$ are in the same cluster, then this is $2\eps$. 
    	\item If $j$ and $i$ are the in the same cluster while $k$ is in the other cluster, this is $2l$.
    \end{enumerate}
    Now, if $j$ and $i$ are in different clusters cluster, then $Q^{j \to i}_c = \eps + l$. Otherwise, if they are in the same cluster, it depends on whether that cluster has size at least $c$. If so, then $Q^{j \to i}_c=2\eps$; otherwise, $Q_c^{j \to i} = 2l$.
    
    Now, let $x$ be an arbitrary node in $X$. If $|X| \ge c$:
    $$
    n\overline{S}^x =\sum_{j} S^{j, x} = \sum_{j \in X} 2\eps + \sum_{j \in Y} \eps + l = |X|2\eps+|Y|(\eps+l)
    $$
    Dividing by $n$ yields the desired result. 
    If $|X| <c$, then $\overline{S}^x = 2\bX l+2\bY (\eps+l)$.
    Since $|X| \ge |Y|$ and $c > 0.5$, an arbitrary node $y \in Y$ only has one case: $\E[\overline{S}^y]=\bX(\eps + l) + 2\bY l$. 
\end{proof}

\begin{lemma} \label{lem:e-cluster}
    In the two cluster model, for $x \in X$ and $y \in Y$, the following holds:
    \[
      \E[\overline{E}^y_m] = \bX(l-\eps)
      \qquad \text{ and } \qquad 
      \E[\overline{E}_m^x] = \begin{cases}
        \bY(\eps -l) & |X| \ge m \\
        \bY(l-\eps)  & \text{otherwise}
      \end{cases}
    \]
\end{lemma}
\begin{proof}
    Consider $\{L_{i \to k} - L_{j \to k}\}_k$. There are three cases:
    \begin{enumerate}
        \item If $i,j$ are in the same cluster, then this is $0$, regardless of which cluster $k$ is in.
        \item If $i, j$ are in different clusters, and $k$ is in $i$'s cluster, then this is $\eps - l$.
        \item If $i, j$ are in different clusters, and $k$ is in $j$'s cluster, then this is $l - \eps$. 
    \end{enumerate}
    So, $E^{i-j}_m = 0$ if $i$ and $j$ are in the same cluster. Otherwise, if $i$'s cluster's size is at least $m$, then $E^{i-j}_m = \eps - l$; if the cluster is smaller, then $E_m^{i-j} = l - \eps$.

    Now, let $x \in X$ be arbitrary. If $|X| \ge m$:
    \begin{align*}
        n\overline{E}^x_m=\sum_{j}E^{x-j}_m = \sum_{j \in X}0 + \sum_{j \in B}\eps - l= |Y|(\eps - l)
    \end{align*}
    If $|X| < m$, then $\overline{E}^x_m = \bY(l-\eps)$.

    Since $|X| \ge |Y|$ and $m > 0.5$, an arbitrary node $y \in Y$ only has one case: $\overline{E}^y_m=\bX(l-\eps)$.     
\end{proof}
\begin{theorem}[cluster,\;$\sigma^*$]\label{thm:fairness-cluster}
    In the cluster model with deterministic latencies with $|X| > c = m$, the advantage is:
    \begin{align*}
        \textsc{Advantage}(\vec{\sigma_\text{early}}) 
        = \frac{2dl + \frac{\mu_0}{k \mu}- \frac{k-1}{k}(2\bX\eps + (1-\bX)(\eps+l))-(1-\bX)(\eps -l)}{2dl+\frac{\mu_0}{k \mu} 
        - \frac{k-1}{k}(\bX(\eps+l) + 2(1-\bX)l)-\bX(l-\eps)}
    \end{align*}
\end{theorem}
\begin{proof}
    First, note that the slowest start time is simply $2l$; thus the timeout $\tau =2ld$. Since $B(\tau) = 0$, we get that $b_0-b\cdot 2dl=0$ or $b_0 = 2bdl$. 
    Further, let $b = k\mu$, so $b_0 = 2k{\mu}dl$.
    
    Recall that 
    \begin{align*}
        \frac{U_x}{U_y} 
        &= \frac{
            (b_0+\mu_0)
            -(b-{\mu})\E[\overline{S}^x]-b\E[\overline{E}_m^x]
            }
            {
            (b_0+\mu_0)
            -(b-\mu)\E[\overline{S}^y]-b\E[\overline{E}_m^y]
            } \\
        &= \frac{
            2k{\mu}dl 
            + \mu_0
            - (k-1) 
                \mu\E[\overline{S}^x]-k{\mu}\E[\overline{E}_m^x]
            }
            {
            2k{\mu}dl 
            + \mu_0
            - (k-1) 
                \mu\E[\overline{S}^y]-k{\mu}\E[\overline{E}_m^y]
            } \\
        &= \frac{
            2dl
            + \frac{\mu_0}{k \mu}
            -\frac{k-1}{k}\E[\overline{S}^x]
            -\E[\overline{E}_m^x]
            }
            {
            2dl
            +\frac{\mu_0}{k \mu}
            - \frac{k-1}{k}\E[\overline{S}^y]-\E[\overline{E}_m^y]
            }
    \end{align*}
    
    We know apply \cref{lem:s-cluster,lem:e-cluster} using the fact that $\bY = 1 - \bX$. 
    The denominator is always $2dl + \frac{\mu_0}{k \mu} - \frac{k-1}{k}(\bX(\eps+l) + 2(1-\bX)l)-\bX(l-\eps)$. 
    For the numerator, it depends on whether the relative magnitudes of $|X|, c$, and $m$. 
    Using the assumption that $|X| > c, m$, the ratio becomes:
    \begin{align*}
        \frac{U_x}{U_y} 
        &= 
        \frac{2dl + \frac{\mu_0}{k \mu}- \frac{k-1}{k}(2\bX\eps + (1-\bX)(\eps+l))-(1-\bX)(\eps -l)}{2dl+\frac{\mu_0}{k \mu} 
        - \frac{k-1}{k}(\bX(\eps+l) + 2(1-\bX)l)-\bX(l-\eps)}
    \end{align*}
\end{proof}

\subsection{Delayed equilibria as a baseline} \label{sec:delayed}

The following lemma will be useful for evaluating the utility ratios in the late-proposing profile.
\begin{lemma} \label{lem:late-ratio}
For the late strategy profile 
$\vec{\sigma_\text{late}}$,
\[
\frac{U_i(\vec{\sigma_\text{late}})}{U_j(\vec{\sigma_\text{late}})}
=\frac{(b_0+\mu_0) + \mu\tau - \mu\E[\overline{E}^{i}_c]}{(b_0+\mu_0) + \mu\tau - \mu\E[\overline{E}^{j}_c]}
\quad \text{where} \quad 
\overline{E}^i_c = \frac{1}{n} \sum_{k \in [n]} E^{i - k}_c
\]
\end{lemma}
\begin{proof}
    First consider the utility ratio of two validators $i$ and $j$.
    Since we are in the late-proposing case, every validator will propose at $\Delta^*$.
    Since $i$ is elected leader preceded by some $j$ with $1/n^2$ probability,
    the expected reward is $\E[\frac{1}{n^2}\sum_j M(\tau - E^{i - j}_c) + \frac{1}{n}b_0]
    = \frac{1}{n}(b_0 + \mu_0 + \mu\tau) - \mu \frac{1}{n^2}\sum_j \E[E^{i - j}_c])$.
    Hence,
    \begin{align*}
        \frac{U_i(\vec{\sigma_\text{late}})}{U_j(\vec{\sigma_\text{late}})}
        &= \frac{(b_0+\mu_0) + \mu\tau - \mu\E[\overline{E}^{i}_c]}{(b_0+\mu_0) + \mu\tau - \mu\E[\overline{E}^{j}_c]}
    \end{align*}
\end{proof}
\subsubsection{Late proposing on a line.} We now look at the line model.
\begin{lemma} \label{lem:fairness-late-line}
For the late strategy profile $\vec{\sigma_\text{late}}$,
the advantage in the deterministic 
line model is given by the following:
    \[
    \textsc{Advantage}(\vec{\sigma_\text{late}})
    = \frac{\frac{b_0+\mu_0}{\mu n} + d - \left( 
            \frac{1}{36} - \frac{2}{3n} \right)}
                {\frac{b_0+\mu_0}{\mu n} + d - \left( 
            \frac{7}{18} - \frac{5}{6n} \right)} 
    \]
\end{lemma}
\begin{proof}
    In the line model, since latencies are proportional to distance, the worst case ratio is then 
    the ratio between the player in the middle and the player on the edge. 
    Note that the longest start time is $S^{0,n-1} = n-1 \approx n$. Hence, $\tau = dn$.
    We also note that as $m=c$ we can use \cref{lemma:q}.
    Now, using \cref{lem:late-ratio},
    \begin{align*}
        \frac{(b_0 + \mu_0) + \mu\tau - \mu\left( 
            \frac{n}{36} - \frac{2}{3} 
            \right)}
             {(b_0 + \mu_0) + \mu\tau - \mu\left( 
            \frac{7n}{18} - \frac{5}{6}  
            \right)}
        &= \frac{\frac{b_0+\mu_0}{n} + \mu d 
            - \mu\left( 
            \frac{1}{36} - \frac{2}{3n} 
            \right)}
                {\frac{b_0+\mu_0}{n} + \mu d - \mu\left( 
            \frac{7}{18} - \frac{5}{6n} \right)} \\
        &= \frac{\frac{b_0+\mu_0}{\mu n} + d - \left( 
            \frac{1}{36} - \frac{2}{3n} \right)}
                {\frac{b_0+\mu_0}{\mu n} + d - \left( 
            \frac{7}{18} - \frac{5}{6n} \right)} 
    \end{align*}
\end{proof}
\subsubsection{Late proposing in clusters.} Now consider two clusters.
\begin{lemma} \label{lem:fairness-late-cluster}
For the late strategy profile $\vec{\sigma_\text{late}}$ with $|X| > c$,
the advantage in the deterministic cluster model is given by the following:
    \[
    \textsc{Advantage}(\vec{\sigma_\text{late}})
    = \frac{(b_0 + \mu_0) + 2dl\mu - \mu ((1 - \bX) (\eps - l))}
                    {(b_0 + \mu_0) + 2dl\mu - \mu (\bX(l - \eps))}
    \]
\end{lemma}
\begin{proof}
    In the cluster model, we know that $|X| > |Y|$ and therefore $x \in X$ will have utility that always dominates 
    $y \in Y$. Note that the slowest start time is also $2l$ so $\tau = 2dl$ and we calculated $\E[\overline{E}^i_c]$ in 
    \cref{lem:e-cluster}. Using \cref{lem:late-ratio},
    \begin{align*}
        \frac{U_x}{U_y} 
        &= \frac{(b_0 + \mu_0) + 2dl\mu - \mu \E[\overline{E}^x_c]}
                {(b_0 + \mu_0) + 2dl\mu - \mu \E[\overline{E}^y_c]}
        = \frac{(b_0 + \mu_0) + 2dl\mu - \mu \E[\overline{E}^x_c]}
                {(b_0 + \mu_0) + 2dl\mu - \mu \bX(l-\eps)}
    \end{align*}
    The value of $\E[\overline{E}^x_c]$ depends on the size of $X$. Since $|X|>c$,
    advantage is:
    \[
    \frac{U_x}{U_y} = \frac{(b_0 + \mu_0) + 2dl\mu - \mu (\bY (\eps - l))}
                    {(b_0 + \mu_0) + 2dl\mu - \mu (\bX(l - \eps))}
    \]
\end{proof}
\begin{theorem}[line and cluster,\;$\sigma_\text{late}$] \label{thm:fairness-late}
    In the deterministic line model, the unfairness can be quantified as, 
    \[
        \textsc{Advantage}(\vec{\sigma_\text{late}}) = 
        \frac{\frac{b_0+\mu_0}{\mu n} + d - \left( 
            \frac{1}{36} - \frac{2}{3n} \right)}
                {\frac{b_0+\mu_0}{\mu n} + d - \left( 
            \frac{7}{18} - \frac{5}{6n} \right)} 
    \]
    and in the deterministic cluster model with $|X| > c$, the unfairness can be quantified as,
    \[
        \textsc{Advantage}(\vec{\sigma_\text{late}}) = \frac{(b_0 + \mu_0) + 2dl\mu - \mu ((1 - \bX) (\eps - l))}
                    {(b_0 + \mu_0) + 2dl\mu - \mu (\bX(l - \eps))}
    \]
\end{theorem}
\begin{proof}
    The result follows directly from \cref{lem:fairness-late-line,lem:fairness-late-cluster}.
\end{proof}

\def\bzero{0.038}   %
\def\mzero{0.005}   %
\def\dpar{5}      %
\def\mpar{6e-6}     %
\def\kpar{1.1}      %
\def\ppar{0.1}      %
\def\Xpar{0.67}      %
\def\xmin{0.05}
\def\xmaxx{100000}
\def\xmax{10000}
\newcommand{\D}[1]{((\bzero+\mzero)/(\mpar*#1))}          %
\newcommand{\B}[1]{(\dpar + \mzero/(\kpar*\mpar*#1))}     %

\pgfplotsset{
  myaxis/.style={
    width=\linewidth, height=5.2cm,
    axis lines=left, grid=both, minor grid style={opacity=.3},
    samples=100,
    scaled y ticks=false,
    yticklabel=\pgfmathprintnumber{\tick}\%, %
    restrict y to domain=-100:100
  }
}
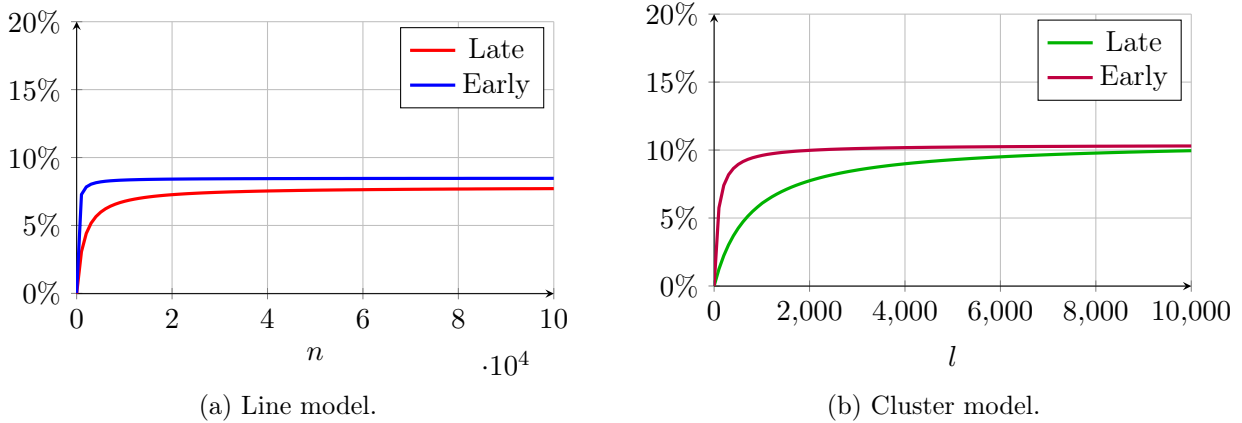
\begin{figure}[htp]
\centering
\begin{subfigure}[t]{0.48\textwidth}
\centering
\begin{tikzpicture}
\begin{axis}[myaxis, domain=\xmin:\xmaxx, xlabel={$n$}, ymax=20]
  \addplot[very thick, red]
    { 100 * (  \D{x} + \dpar - ( 1/36 + 2/(3*x) )  )
      / (  \D{x} + \dpar - ( 7/18  - 5/(6*x) ) ) - 100 };
  \addlegendentry{Late}
  \addplot[very thick, blue]
    { 100 * (  \B{x} - ((\kpar-1)/\kpar)*(25/36 - 2/(3*x)) - 1/36 - 2/(3*x) )
      / (  \B{x} - ((\kpar-1)/\kpar)*(17/18 - 7/(6*x)) - 7/18 - 5/(6*x) )
      - 100 };
  \addlegendentry{Early}
\end{axis}
\end{tikzpicture}
\caption{Line model.}
\end{subfigure}
\hfill
\begin{subfigure}[t]{0.48\textwidth}
\centering
\begin{tikzpicture}
\begin{axis}[myaxis, domain=\xmin:\xmax, xlabel={$l$}, ymax=20]
  \addplot[very thick, green!70!black]
    { 100 * ( (\bzero+\mzero) + 2*\dpar*x*\mpar - \mpar*((1-\Xpar)*(\ppar - x)) )
      / ( (\bzero+\mzero) + 2*\dpar*x*\mpar - \mpar*(\Xpar*(x - \ppar)) )
      - 100 };
  \addlegendentry{Late}
  \addplot[very thick, purple]
    { 100 * ( 2*\dpar*x + \mzero/(\kpar*\mpar)
        - ((\kpar-1)/\kpar)*( 2*\Xpar*\ppar + (1-\Xpar)*(\ppar + x) )
        - (1-\Xpar)*(\ppar - x) )
      / ( 2*\dpar*x + \mzero/(\kpar*\mpar)
        - ((\kpar-1)/\kpar)*( \Xpar*(\ppar - x) + 2*(1-\Xpar)*x )
        - \Xpar*(x - \ppar) )
      - 100 };
  \addlegendentry{Early}
\end{axis}
\end{tikzpicture}
\caption{Cluster model.}
\end{subfigure}
\caption{Advantage comparison between early and late strategy profiles. Plotted expressions are from \cref{thm:fairness-line,thm:fairness-cluster,thm:fairness-late}. The parameters are set such that $\mu_0 = 0.005, b_0 = 0.038, d=5,\mu=6\times10^{-6},k=1.5,\eps=0.1,\bX=0.67$.}
 \label{fig:fairness-plots}
\end{figure}

    \section{Simulation details} \label{appendix:sim}
\ifdeferproofs \crefalias{appendix:sim}{appendix} \fi

\subsubsection{World latency simulation.}
For the world latency model, we use a lognormal where the underlying normal is distributed with 
mean $1$ and std $0.5$ within cities. Between cities, we fix a std of $0.8$
and the log of the table below as the mean.\footnote{Taken from \url{https://wondernetwork.com/pings/} on 13-08-2025.}
\sisetup{
  table-number-alignment = center,
  table-figures-integer = 3,
  table-figures-decimal = 3
}

\begin{table}[H]
    \centering
    \caption{List of cities for the world map simulation.}
    \begin{tabularx}{\textwidth}{@{}*{5}{>{\centering\arraybackslash}X}@{}}
    Amsterdam & Dublin & Frankfurt & Los Angeles & New York \\
    Paris     & Seoul  & Sydney    & Vienna      & Warsaw   \\
    Zurich    &        &           &             &          \\
    \end{tabularx}
\end{table}

\newcommand{\city}[1]{\multicolumn{1}{c}{#1}}
\newcolumntype{M}{S[table-format=3.2,round-mode=places,round-precision=2]<{\;\si{ms}}}
\begin{table}[H]
\centering
\caption{Inter-city latency (ms)}
\scriptsize
\setlength{\tabcolsep}{8pt}
\begin{tabular}{l *{6}{M}}
\toprule
& \city{Amsterdam} & \city{Dublin} & \city{Frankfurt} & \city{Los~Angeles} & \city{New~York} & \city{Paris} \\
\midrule
Amsterdam   & 0.0    & 18.898 & 8.758  & 151.911 & 77.245 & 35.695 \\
Dublin      & 18.929 & 0.0    & 22.573 & 148.112 & 79.363 & 17.038 \\
Frankfurt   & 8.788  & 22.438 & 0.0    & 157.842 & 83.328 & 15.538 \\
Los~Angeles & 151.8  & 148.003 & 157.759 & 0.0    & 68.998 & 144.772 \\
New~York    & 77.233 & 79.835 & 83.573 & 68.949 & 0.0    & 72.738 \\
Paris       & 35.491 & 17.690 & 15.341 & 144.803 & 72.775 & 0.0    \\
Seoul       & 216.865 & 248.608 & 222.590 & 134.224 & 202.454 & 310.963 \\
Sydney      & 252.216 & 258.316 & 267.893 & 160.498 & 200.354 & 242.455 \\
Vienna      & 18.625 & 33.620 & 24.064 & 168.068 & 94.419 & 38.412 \\
Warsaw      & 22.942 & 43.371 & 22.771 & 171.280 & 95.172 & 33.041 \\
Zurich      & 16.290 & 35.949 & 7.004  & 159.108 & 104.166 & 18.162 \\
\bottomrule
\end{tabular}
\end{table}
\begin{table}[H]
\centering
\caption{Inter-city latency (ms)}
\scriptsize
\setlength{\tabcolsep}{8pt}
\begin{tabular}{l *{5}{M} S[table-format=2.0]}
\toprule
& \city{Seoul} & \city{Sydney} & \city{Vienna} & \city{Warsaw} & \city{Zurich} & {$w_i$} \\
\midrule
Amsterdam   & 216.451 & 252.287 & 18.933 & 23.054 & 16.311 & 18 \\
Dublin      & 248.549 & 259.371 & 33.799 & 43.349 & 35.984 & 17 \\
Frankfurt   & 222.781 & 267.780 & 24.022 & 22.531 & 7.014  & 17 \\
Los~Angeles & 134.145 & 160.565 & 168.401 & 171.428 & 159.026 & 8 \\
New~York    & 202.481 & 200.356 & 95.397 & 95.097 & 104.104 & 8 \\
Paris       & 309.104 & 242.421 & 38.392 & 33.004 & 18.147 & 11 \\
Seoul       & 0.0     & 137.450 & 235.099 & 237.825 & 259.888 & 5 \\
Sydney      & 137.553 & 0.0     & 260.271 & 275.894 & 276.444 & 6 \\
Vienna      & 235.010 & 260.584 & 0.0    & 27.633 & 19.382 & 7 \\
Warsaw      & 237.683 & 275.929 & 27.679 & 0.0    & 23.273 & 3 \\
Zurich      & 259.767 & 276.520 & 19.421 & 23.304 & 0.0    & 15 \\
\bottomrule
\end{tabular}
\end{table}

The weight indicates how many nodes are in that city out of 100.

    \section{Inflation and griefing incentives} \label{appendix:inflation}
In our standard model, the block reward is \textit{printed} by the protocol, and thus is a form of inflation. Thus, there is an incentive to lower others' rewards, as that increases the value of one's own rewards. We refer to this as a \textit{griefing} incentive. Here, we argue that this incentive would not significantly change the results.

First, for the Nash equilibrium without coalitions, note that a validator can lower other validator's rewards, on average, by always timing out (or voting infinity). This would lower the leader's reward if and only if the dishonest validator was one of the $m$ earlier validators to receive the leader's proposal. However, with large $n$, any one validator would have a very small impact on the aggregate measurement -- the leader would receive the $(m+1)$-smallest vote instead of the $m$-smallest vote. So, adding inflation to our model would result in an $\eps$-BNE here, where $\eps$ would be very small. 

Second, small coalitions would have a larger ability to affect others' rewards. If they vote $0$ on coalition leaders and vote $\infty$ for non-coalition leaders, they would lower the honest leaders' reward from the $m$th smallest vote to the $(m+|C|)$th smallest vote (among honest validators). This would shift rewards from non-coalition members to coalition members. However, the optimal proposal strategy, for both coalition and non-coalition leaders, is still to propose early -- this follows directly from the time-decreasing property. So, responsiveness is still intact.

Third, large coalitions would be able to completely censor non-coalition leaders. While this is problematic, we note that the same problem occurs in existing system (i.e., with static block rewards). It would always be in a large coalition's interest to deny rewards to non-coalition members, whether they are dynamic or static.

Finally, the fairness ratios would be unchanged, as the net inflation would devalue all profits by the same multiplicative factor.

    \section{Alternative incentive knob: changing leader election} 
\label{apx:leader} 
\ifdeferproofs \crefalias{apx:leader}{appendix} \fi
We now consider varying the other incentive knob: the leader schedule.
We think that this alternative might be of practical interest for 
ease of implementation.

\subsubsection{Incentive knobs.}
As we have discussed so far, once the timeliness of a protocol is measured, there are two 
aspects of the protocol that can be used by the mechanism designer to set the incentives. 
Both of these result in a rich incentive structure that we analyze separately in this paper.
\begin{itemize}
    \item \emph{Block reward:} The block reward can be a decreasing function of the proposal duration. 
    Intuitively, if the initial reward is high enough compared to MEV gain, the incentives shift to 
    early proposal.
    \item \emph{Leader probabilities:} We explore designing a leader election schedule such that for each 
    player, the later the proposal is, the next leader election probability is proportionally reduced. 
    The result is a direct reduction in future rewards with the goal of incentivizing early proposals.
\end{itemize}

\subsubsection{Leader weights.} Each player now has a weight variable $w_i$.
Intuitively, our goal is to reduce the weight of player $i$ proportional to how 
late it was to propose in its previous proposal.
Initially, the weights are equal to 1 so that each player is selected with probability 
$1/n$ (Note that our model normalizes by stake, and consider coalitions to account for stake distributions).
After each proposal, we update the leader's weight as a function of it's measured proposal duration
using a weight function $\mathcal{W} : [0, \tau] \to [0,1]$.
Then, the election probability for each $i$ is $w_i/\sum_j w_j$.

\subsubsection{Stake transfer.} We do not explicitly model the mechanisms that allow exiting and 
entering the system. Clearly, if a staked entity is reduced to 0 weight, it is not permanent 
as they are able to withdraw their stake and re-stake at a new keypair to start over.
In practice, implementations could choose to start new stake at lower weight and 
have a delay on withdrawals in order to  
prevent incentives to propose late and avoid punishment by withdrawing.

\subsubsection{Choosing the weight function.}
While the design space for selecting the weight function is quite large, we would like 
the weight function to not introduce strong incentives to run more validators. 
Intuitively, it shouldn't matter if you are running a single validator with more stake 
compared to lots of validators with the minimum amount of stake. Since our model 
considers each entity as a single unit of stake anyway, we capture this notion by 
designing the weight function to apply an \emph{additive} effect on previous weight.
Hence, no matter how large your group is as a coalition, for an action $\Delta$, 
you get the same effect of $+ \W(\Delta)$, effectively reducing the consideration to 
total weight.

\newcommand{\changed}[1]{\textbf{#1}}
\subsubsection{The timing game with leader decay.} The changes from the timing game in \cref{sec:model}
are in bold.
\begin{gamebox}[]{timing game with leader decay}
\begin{enumerate}
    \item Let $R$ be the current round, and suppose $i$ is the current leader and $j$ led the previous round. Let $t_{R-1}$ represent the time when round $R-1$ ended, i.e., when $j$ broadcasted their block. Let $s^{R-1}$ denote the time it takes between when $j$ broadcasts their block and when $i$ receives the block and its quorum. $i$ learns their maximum delay $\Delta^*_R$ and chooses a delay $\Delta \in [0, \Delta^*_R]$.
    \item Let $k$ be an arbitrary validator. 
        \begin{enumerate}
            \item The latency $l_{i \to k} \sim L_{i \to k}$ is sampled.
            \item $k$ receives the block at time $t_{k}^R = t_{R-1} + s^{R-1}+ \Delta + l_{i \to k}$. 
            \item $k$ votes according to their voting strategy, which is a function of the block times they have seen. So $v_k^R =\sigma^V_k(\mathbf{t}_k^{1:R}) \in \mathbb{R}^+$. $k$ then broadcasts\footnote{Depending on the protocol, $k$ may know the next leader and send their vote only to them, but this detail is not relevant in our model.} their vote.\footnote{As discussed earlier, in practice, one would decouple consensus from the duration voting, by only voting ``yes'', ``no'', or ``timeout'' in the critical path, and submitting duration votes on-chain for later aggregation. For the purposes of this model, we do not make this distinction.} 
        \end{enumerate}
    \item Let $v_m^*$ refer to the $m\ts{th}$ smallest vote in the 
    (multi)set $\{v_k^R\}_k$. $i$ receives utility $M(\Delta+s^{R-1}) 
    + B(v_m^*)$, where $M$ denotes the MEV and $B$ the block reward. 
    \item \changed{Update the weight of the leader by \boldmath$w_{i} := w_{i} + \W(v_m^*)$.}
    \item The next leader is randomly sampled \changed{from the distribution \boldmath$\{w_j / \sum_k w_k \}_j$}
    and the game repeats from 1.
\end{enumerate}
The first round of this game is slightly different. 
\changed{Initialize weights \boldmath$w_i = 1$ for all \boldmath$i \in [n]$.}
The game starts at time $t = 0$, and so $s^{0} = 0$, i.e., the leader $i$ can propose as soon as the game starts. 
As a result, validator $k$ receives the block at time $t_k^1 = \Delta + l_{i \to k}$. Otherwise, the first round matches the others.
\end{gamebox}
Let $\sigma^*$ denote the strategy profile, in which players report honestly,
each leader $i$ proposes with $\Delta=0$. 
Similar to last time, we refer to $\sigma^*$ as the \emph{early-proposing truthful} strategy profile.

\subsection{Equilibria under binary state}

As an instructive example, consider a step function as leader weight that oscillates weights between high weight
and low weight. Nodes never go below $1-\rho$, and above $1$. If the step function is simply $\pm \rho$ 
we get a two state Markov process to analyze.
Formally, let $\mathcal{T}$ be a target time such that, intuitively, most quorums can be formed before $\mathcal{T}$.
Let $\rho$ be weight function parameter such that if people propose before $\mathcal{T}$, we apply $+\rho$,
and anything later than $\mathcal{T}$, we apply $-\rho$.
We call this the \emph{binary weight} model.

We will begin by analyzing the early-proposing truthful strategy $\sigma^*$.
A key probability due to latencies is the probability that the reported delay is above $\mathcal{T}$,
hence resulting in a penalty. For player $i$, we will represent this probability with
$p_i := \Pr[S^{j,i} + E^{i-j}_m > \mathcal{T}]$.
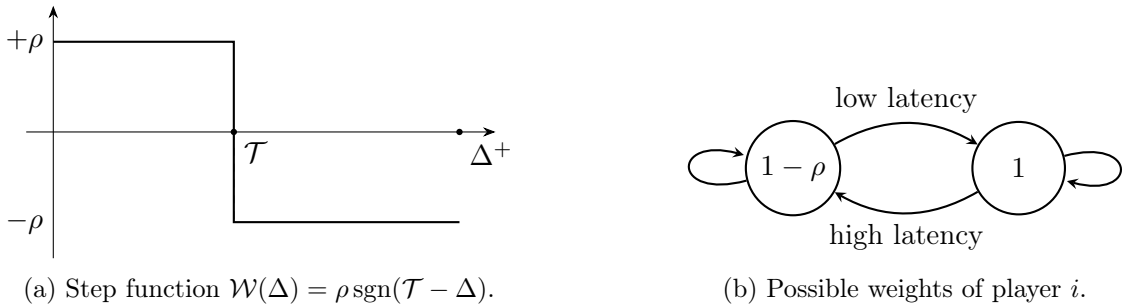
\begin{figure}[htbp]
  \centering
  \begin{subfigure}[b]{0.48\textwidth}
    \centering
    \begin{tikzpicture}[scale=1.2, > = Stealth]
      \def\rhoval{1.0}   %
      \def\tauloc{2.0}   %
      \def\L{4.5}        %
      \draw[->] (-0.3,0) -- (\L+0.4,0) node[below] {};
      \draw[->] (0,-1.4*\rhoval) -- (0,1.4*\rhoval) node[left] {};
      \draw[thick] 
            (0,\rhoval) -- (\tauloc,\rhoval) -- (\tauloc,-\rhoval) -- (\L,-\rhoval);
      \node[left]  at (0,\rhoval)  {$+\rho$};
      \node[left]  at (0,-\rhoval) {$-\rho$};
      \draw[dashed] (\tauloc,0) -- (\tauloc,-\rhoval);
      \fill (\tauloc,0) circle (1pt) node[below right] {$\mathcal{T}$};
      \fill (\L,0) circle (1pt) node[below right] {$\Delta^{+}$};
    \end{tikzpicture}
    \caption{Step function $\W(\Delta)=\rho\,\mathrm{sgn}(\mathcal{T}-\Delta)$.}
    \label{fig:step}
  \end{subfigure}
  \hfill
  \begin{subfigure}[b]{0.48\textwidth}
    \centering
    \begin{tikzpicture}[>=stealth, node distance=3cm,
                        every state/.style={thick, minimum size=35pt}]
      \node[state] (q0) {$1 - \rho$};
      \node[state, right of=q0] (q1) {$1$};
      \path[->, thick] (q0) edge[bend left]  node[above] {low latency} (q1)
                            edge[loop left]  node[left]  {} (q0)
                       (q1) edge[bend left]  node[below] {high latency} (q0)
                            edge[loop right] node[right] {} (q1);
    \end{tikzpicture}
    \caption{Possible weights of player $i$.}
    \label{fig:automaton}
  \end{subfigure}

  \caption{Binary weight model.}
  \label{fig:step+automaton}
\end{figure}

We first observe that once we fix the strategy to $\sigma^*$, the stationary probabilities of 
any player $i$ having weight $w_i$ equal to $1-\rho$ or $1$ is 
independent of the latency distributions of 
other players.
For analysis, consider the following equivalent formulation and the induced Markov chain
$M^{\sigma^*}$ with two states, $1$ and $1-\rho$. In each state, first uniformly randomly pick
some player. If this player is $i$, toss a coin with probability $w_i$. If $i$ survives this coin, 
toss a last coin with probability $p_i$ to determine whether to transition to $1-\rho$ or $1$.
On all other cases, stay in the same state, essentially skipping $i$.
\begin{lemmarep} \label{lem:stationary-weight}
    In the binary weight model with players adopting $\sigma^*$, 
    the stationary distribution of weights in $M^{\sigma^*}$ 
    for any player $i$ is the following:
    \[
    \nu_{1-\rho} = \frac{p_i}{1-\rho(1-p_i)}
    \qquad \text{ and } \qquad 
    \nu_{1} = \frac{(1-\rho)(1-p_i)}{1-\rho(1-p_i)}
    \]
\end{lemmarep}
\begin{proof}
    Observe that the probability of going from $1$ to $1-\rho$ requires $i$ to be selected 
    with probability proportional to its weight $1$ and having high latency.
    Similarly, from $1-\rho$ to $1$, $i$ gets elected with probability proportional to $1-\rho$ and 
    has low latency. Therefore, we can compute the stationary probabilities by:
    \[
    \nu_{1-\rho} = \frac{p_i}{p_i + (1-\rho)(1-p_i)}  = \frac{p_i}{1-\rho(1-p_i)}
    \]
    and 
    \[
    \nu_1= \frac{(1-\rho)(1-p_i)}{p_i + (1-\rho)(1-p_i)} = \frac{(1-\rho)(1-p_i)}{1-\rho(1-p_i)}
    \]
    Note that we omit the factor of $1/n$ as it cancels out.
\end{proof}

Unfortunately, however, there exists an explicit incentive to deviate from $\sigma^*$. 
Consider the scenario in which player $j$ is leaders and during aggregation the latencies were 
borderline---slightly aggregating to below $\mathcal{T}$ with $i$ having the power to push it above $\mathcal{T}$.
In this case, if player $i$ lies, $j$'s weight decreases, increasing $i$'s probability of winning and hence 
utility.

\subsubsection{Low latency.} We do however recover equilibria if $i$ never has the potential to cast a deciding vote.
Consider the scenario in which even the $(m+1)\ts{th}$ vote is always guaranteed to be less than $\mathcal{T}$.
If no reported duration is ever above $\mathcal{T}$, then the players always have weight 1.
\begin{corollaryrep} \label{cor:nash-leader-decay}
    If $\Pr[S^{j,i} + E^{i-j}_{m+1} > \mathcal{T}] = 0$ for all $i$ and $j$ and $m \geq n-c$, 
    then each player $i$ has the following utility:
    \[
    U_i(\sigma^*)
    =
        \frac{\sum_j\E[M(S^{j, i})] + \E\left[B\left(S^{j, i}+ E^{i-j}_{m}\right)\right]}{\sum_{i, j} \E[S^{j, i}]}
    \]
    Moreover, if $\langle M,B \rangle$ are time-decreasing, $\sigma^*$ is a subgame-perfect Bayes-Nash equilibria.
\end{corollaryrep}
\begin{proof}
    Observe that if the assumption holds, \cref{lem:stationary-weight} implies that all players stay 
    at weight 1, even if some player $i$ votes dishonestly, as at best it can influence the $(m+1)\ts{th}$ 
    vote but not the $m\ts{th}$. Hence, we recover the case in \cref{thm:latency-NE}.
\end{proof}
This section highlights the fact that the equilibria structure after introducing decaying 
leader election probabilities depends highly on the latency distributions and how the decay mechanism is designed.

\end{toappendix}

\begin{toappendix}
    \section{List of symbols} \label{apx:glossary}
    {\renewcommand*{\glossarysection}[2][]{}
    \printunsrtglossary[type=symbols,style=long]}
\end{toappendix}

\bibliography{ref.bib}
\end{document}